\documentclass[a4paper, 11pt]{article}

\usepackage{listings}
\title{
Neural networks as fuzzy logic formulas
} 

\date{} 					

\usepackage{authblk}

\author[1]{Damian Heiman}
\author[1]{Antti Kuusisto}
\author[1]{Esko Turunen}
\affil[1]{Mathematics Research Centre, Tampere University, Tampere, Finland}

\usepackage{graphicx} 

 \usepackage[textwidth=151mm, hcentering, vmargin={24mm,32.5mm}, foot=20mm]{geometry}

\usepackage{verbatim}

\usepackage{geometry}
\usepackage{graphicx}
\usepackage{hyperref}
\usepackage{fontawesome5}

\usepackage[utf8]{inputenc}
\usepackage[T1]{fontenc}
\usepackage{amsmath, amssymb, amsthm}
\usepackage{tikz}
\usepackage{lipsum}
\usepackage{pdfpages}
\usepackage{multicol} 
\usepackage{stmaryrd}

\usepackage{thm-restate}
\usepackage{hyperref}

\usepackage{makecell}
\usepackage{mathtools}
\usepackage{colonequals}

\usepackage{marginnote}

\usepackage{enumitem}
\usepackage{xspace}

\usepackage{bm}

\theoremstyle{plain}
\newtheorem{theorem}{Theorem}[section]
\newtheorem{lemma}[theorem]{Lemma}
\newtheorem{proposition}[theorem]{Proposition}
\newtheorem{corollary}[theorem]{Corollary}

\theoremstyle{definition}
\newtheorem{remark}[theorem]{Remark}
\newtheorem{example}[theorem]{Example}
\newtheorem{definition}[theorem]{Definition}

\newcommand{\N}{\mathbb N}
\newcommand{\Z}{\mathbb Z}
\newcommand{\Q}{\mathbb Q}
\newcommand{\R}{\mathbb R}

\newcommand{\cL}{\mathcal{L}}

\newcommand{\cT}{\mathcal{T}}

\newcommand{\id}{\operatorname{id}}

\newcommand{\LPH}{\mathit{L \Pi} \frac{1}{2}}
\newcommand{\RPL}{\mathrm{RPL}}
\newcommand{\RPLb}{\mathrm{PL}}
\newcommand{\RPLM}{\RPL(\odot)}
\newcommand{\RPLMb}{\mathrm{PL}(\odot)}
\newcommand{\PROP}{\mathrm{PROP}}
\newcommand{\VAR}{\mathrm{VAR}}
\newcommand{\imp}{\rightarrow}
\newcommand{\impl}{\imp_{L}}
\newcommand{\impp}{\imp_{P}}

\newcommand{\negl}{\neg_{L}}
\newcommand{\equivl}{\leftrightarrow_{L}}
\newcommand{\half}{\frac{1}{2}}

\newcommand{\ReLU}{\mathrm{ReLU}}

\newcommand{\scale}{\mathrm{scale}}
\newcommand{\lmodel}{V}
\newcommand{\tmodel}{E}

\usepackage{array,multirow}

\newcommand{\defstyle}[1]{\textbf{#1}}

\newcommand{\pand}{\otimes}

\newcommand{\length}{\mathrm{length}}
\newcommand{\subt}{\mathrm{subt}}

\newcommand{\LPHF}{\LPH(\impp^-)}
\newcommand{\LPHFF}{\LPH(\odot^-, \impp^-)}
\newcommand{\coolname}{$\Q[\ReLU, x_i]_{i \in \N}$\xspace}
\newcommand{\coolnameb}{$\R[\ReLU, x_i]_{i \in \N}$\xspace}
\newcommand{\itsum}{\bigodot}

\begin{document}

\maketitle

\begin{abstract}
    Neural networks are a fundamental aspect of modern artificial intelligence, playing a key role in various important machine learning architectures including transformers and graph neural networks. 
    Recently, logical characterisations have been used to study the expressive power of many machine learning architectures, but logical characterisations of plain neural networks have received less attention.
    In this paper, we provide fuzzy logic characterisations of rational-weight $\ReLU$-activated neural networks via
    Rational Pavelka logic ($\RPL$) and an extension of $\RPL$ called $\RPLM_{\leq 1}$, as well as two fragments of $\LPH$ called $\LPHF_{\leq 1}$ and $\LPHFF$.
    The activation values of the neural networks are allowed to be arbitrary real numbers. 
    We also provide fuzzy logic characterisations of a generalised polynomial ring over $\Q$ in countably many variables where the use of the $\ReLU$-function is permitted 
    via the fuzzy logic $\RPLM$ and a fragment of $\LPH$ called $\LPHF$.
\end{abstract}

\section{Introduction}\label{section: introduction}

Neural networks are directed graphs where each vertex (or neuron) obtains an activation value (a real number) either as an input or by aggregating the activation values of its in-neighbours. 
The aggregation scheme we consider is a standard one, i.e., we take a weighted sum of the activation values of in-neighbours, add a bias and apply the activation function 
$\ReLU(x) = \max\{0,x\}$. 
The neural networks we consider are feedforward (i.e., there are no directed cycles in the graph) and fully connected (i.e., the neurons are partitioned into layers, and from each neuron of a layer there is an edge to exactly every neuron of the next layer), and the weights and biases are rational numbers. An illustration of a neural network is given in Figure~\ref{figure: neural network}.

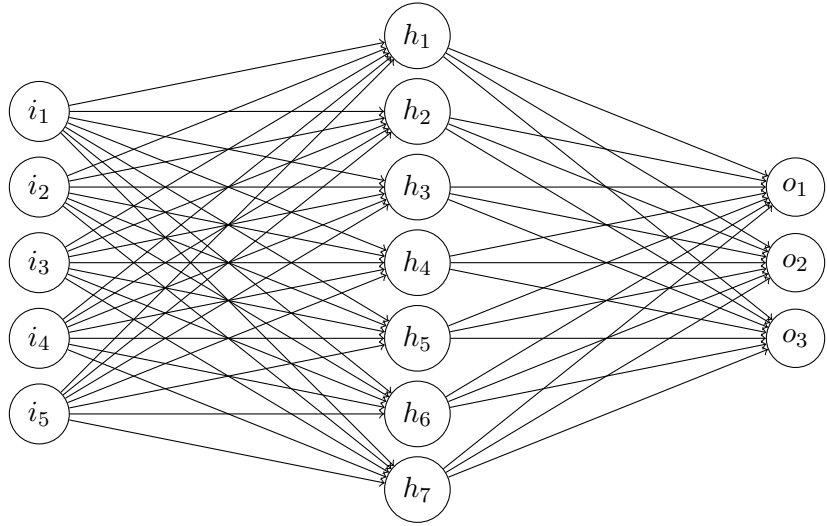
\begin{figure}
    \begin{center}
        \begin{tikzpicture}
            \node[draw, circle] at (0,1) (11) {$i_5$};
            \node[draw, circle] at (0,2) (12) {$i_4$};
            \node[draw, circle] at (0,3) (13) {$i_3$};
            \node[draw, circle] at (0,4) (14) {$i_2$};
            \node[draw, circle] at (0,5) (15) {$i_1$};

            \node[draw, circle] at (5,0) (21) {$h_7$};
            \node[draw, circle] at (5,1) (22) {$h_6$};
            \node[draw, circle] at (5,2) (23) {$h_5$};
            \node[draw, circle] at (5,3) (24) {$h_4$};
            \node[draw, circle] at (5,4) (25) {$h_3$};
            \node[draw, circle] at (5,5) (26) {$h_2$};
            \node[draw, circle] at (5,6) (27) {$h_1$};

            \node[draw, circle] at (10,2) (31) {$o_3$};
            \node[draw, circle] at (10,3) (32) {$o_2$};
            \node[draw, circle] at (10,4) (33) {$o_1$};

            \draw[->] (11) to (21);
            \draw[->] (11) to (22);
            \draw[->] (11) to (23);
            \draw[->] (11) to (24);
            \draw[->] (11) to (25);
            \draw[->] (11) to (26);
            \draw[->] (11) to (27);

            \draw[->] (12) to (21);
            \draw[->] (12) to (22);
            \draw[->] (12) to (23);
            \draw[->] (12) to (24);
            \draw[->] (12) to (25);
            \draw[->] (12) to (26);
            \draw[->] (12) to (27);

            \draw[->] (13) to (21);
            \draw[->] (13) to (22);
            \draw[->] (13) to (23);
            \draw[->] (13) to (24);
            \draw[->] (13) to (25);
            \draw[->] (13) to (26);
            \draw[->] (13) to (27);

            \draw[->] (14) to (21);
            \draw[->] (14) to (22);
            \draw[->] (14) to (23);
            \draw[->] (14) to (24);
            \draw[->] (14) to (25);
            \draw[->] (14) to (26);
            \draw[->] (14) to (27);

            \draw[->] (15) to (21);
            \draw[->] (15) to (22);
            \draw[->] (15) to (23);
            \draw[->] (15) to (24);
            \draw[->] (15) to (25);
            \draw[->] (15) to (26);
            \draw[->] (15) to (27);

            \draw[->] (21) to (31);
            \draw[->] (21) to (32);
            \draw[->] (21) to (33);

            \draw[->] (22) to (31);
            \draw[->] (22) to (32);
            \draw[->] (22) to (33);

            \draw[->] (23) to (31);
            \draw[->] (23) to (32);
            \draw[->] (23) to (33);

            \draw[->] (24) to (31);
            \draw[->] (24) to (32);
            \draw[->] (24) to (33);

            \draw[->] (25) to (31);
            \draw[->] (25) to (32);
            \draw[->] (25) to (33);

            \draw[->] (26) to (31);
            \draw[->] (26) to (32);
            \draw[->] (26) to (33);

            \draw[->] (27) to (31);
            \draw[->] (27) to (32);
            \draw[->] (27) to (33);
            
        \end{tikzpicture}
    \end{center}
    \caption{A visualization of a neural network consisting of three layers: an \emph{input layer} $(i_1, \dots, i_5)$, a \emph{hidden layer} $(h_1, \dots, h_7)$ and an \emph{output layer} $(o_1, o_2, o_3)$.}
    \label{figure: neural network}
\end{figure}

On the formal level, we identify a neuron with its aggregation formula broken down recursively to the 
neurons on the first layer
and we subsequently identify a neural network with the tuple of 
neurons on the last layer.
This makes neurons a special case of polynomials of the polynomial ring $\Q[x_i]_{i \in \N}$ (i.e., the ring of polynomials over rational numbers in countably many variables) extended with the $\ReLU$-function. We refer to these $\ReLU$-polynomials as terms, and we call the collection of these terms \coolname.
Neural networks are a subclass of tuples of terms of \coolname.

We study the expressive power of \coolname and of neural networks via two fuzzy logics that can express rational numbers and multiplication. Fuzzy logics generalise ordinary Boolean logics by allowing the truth value of a formula to be any value in the closed interval $[0,1]$. The first 
fuzzy logic we consider
is the logic $\RPLM$ \cite{Hajek98}, which is an extension of Rational Pavelka Logic ($\RPL$) \cite{hajek1995fuzzy, Pavelka79a, Pavelka79b, Pavelka79c} with the product conjunction~$\odot$, while $\RPL$ extends the classic infinite-valued {\L}ukasiewicz logic \cite{lukasiewicz1930} by adding a truth constant for every rational number between $0$ and $1$.
{\L}ukasiewicz logic is simply a fuzzy version of Boolean propositional logic.
The other fuzzy logic we consider is the logic $\LPH$ \cite{esteva2001} which combines standard {\L}ukasiewicz logic with Product Logic (adding the product conjunction~$\odot$ and the product implication $\impp$) and adds the truth constant $\overline{\half}$. 
The logic $\LPH$ allows the construction of formulae equivalent to rational numbers between $0$ and $1$ by using the product implication $\impp$ which corresponds to division. However, the connective $\impp$ also increases the expressive power of $\LPH$ beyond that of $\RPLM$. We therefore restrict to the fragment $\LPHF$ which limits $\LPH$ by not allowing proposition symbols in the scope of~$\impp$. 

As the values of terms of \coolname and the activation values of neurons are allowed to be arbitrary reals and the truth values of fuzzy logic lie between $0$ and $1$, it is obviously impossible in the general case to find a formula of fuzzy logic whose truth values would be identical to the values of a given term of \coolname. Therefore, we naturally scale the activation values down to the interval $[0,1]$ when translating to the logics. This is the most technical part of the translation, as multiplication behaves differently in the interval $[0,1]$ compared to arbitrary intervals $[-k,k]$. 
No scaling is necessary in the opposite direction.
In our characterisations, we translate tuples of terms to tuples of formulae and vice versa, rather than translating terms to formulae and vice versa.
Our first theorem is as follows:

\begin{restatable*}{theorem}{firsttheorem}\label{theorem: R-algebra = RPLM = LPHF}
    \coolname has the same expressive power (w.r.t. scaling) as $\RPLM$ and $\LPHF$.
\end{restatable*}

Theorem~\ref{theorem: R-algebra = RPLM = LPHF} also holds when we restrict the degrees of the $\ReLU$-polynomials and apply a corresponding restriction on the logic side. This involves defining fragments $\RPLM_{\leq d}$ and $\LPHF_{\leq d}$ of $\RPLM$ and $\LPHF$, respectively, which limit how proposition symbols are allowed to appear in the scope of the product conjunction $\odot$. We obtain the following result:

\begin{restatable*}{theorem}{optionaltheorem}\label{theorem: coolname with degree bound = RPLM and LPHF with deree bound}
    For each $d \in \N$, \coolname with degree bound $d$ has the same expressive power (w.r.t. scaling) as $\RPLM_{\leq d}$ and $\LPHF_{\leq d}$.
\end{restatable*}

When characterising neural networks, we consider on the one hand the fragments $\RPLM_{\leq 1}$ and $\LPHF_{\leq 1}$, which restrict the respective logics $\RPLM$ and $\LPHF$ by demanding that all proposition symbols in the scope of a product conjunction $\odot$ appear on the same side of the connective. On the other hand, we also consider $\RPL$ and $\LPHFF$, where $\LPHFF$ limits both $\odot$ and $\impp$ by not allowing proposition symbols in the scope of either connective. We find that we can translate $\RPLM_{\leq 1}$ and $\LPHF_{\leq 1}$ into neural networks and that we can translate neural networks into $\RPL$ and $\LPHFF$. 
Our third result is then given as follows:

\begin{restatable*}{theorem}{secondtheorem}\label{theorem: neural networks = RPLM1 = LPHF1}
    Neural networks have the same expressive power (w.r.t. scaling) as the logics $\RPLM_{\leq 1}$, $\LPHF_{\leq 1}$, $\RPL$ and $\LPHFF$.
\end{restatable*}

Note that in the results listed above, the notion of expressive power is tied to scaling; it does \emph{not} mean that we could, e.g., translate each formula $\varphi$ of $\RPLM_{\leq 1}$ into a formula~$\varphi'$ of $\RPL$ such that $\varphi'$ would always receive the same truth value as $\varphi$. However, it \emph{does} follow that we could obtain a formula $\varphi'$ of $\RPL$ that arithmetically simulates the truth values obtained by $\varphi$.
A summary of our fuzzy logic characterisations is given in Table~\ref{table: result summary}.
In Appendix~\ref{appendix: lukasiewicz logic}, we discuss why no analogous characterisation can be given via infinite-valued {\L}ukasiewicz logic.

\begin{table*}
\caption{A summary of our results. The symbol $\equiv$ denotes same expressive power (subject to scaling when translating from left to right).}
\begin{center}
\def\arraystretch{1.5}
\begin{tabular}{|c|}
    \hline
    \coolname $\equiv$ $\RPLM$, $\LPHF$ \\
    \hline
    \coolname with degree bound $d$ $\equiv$ $\RPLM_{\leq d}$, $\LPHF_{\leq d}$ \\
    \hline
    Neural networks $\equiv$ $\RPLM_{\leq 1}$, $\LPHF_{\leq 1}$, $\RPL$, $\LPHFF$ \\
    \hline
\end{tabular}
\end{center}
\label{table: result summary}
\end{table*}

\textbf{Related work}

Past fuzzy logic characterisations of neural networks have been obtained via representation theorems, showing that a given extension of {\L}ukasiewicz logic is capable of expressing exactly all McNaughton functions with appropriate restrictions placed on the coefficients. For example, a well-known result by McNaughton \cite{mcnaughton1951} states that plain {\L}ukasiewicz logic can express exactly all McNaughton functions with integer coefficients, linking {\L}ukasiewicz logic to neural networks that use the truncated identity $f(x) = \min\{1, \max\{0,x\}\}$ (also known as the truncated $\ReLU$) as the activation function. 
Amato, Di Nola and Gerla~\cite{amato2005} gave a fuzzy logic characterisation for rational-parameter neural networks via Rational {\L}ukasiewicz logic via McNaughton functions with rational coefficients. An analogous result was given by Di Nola, Gerla and Leustean in \cite{di2013}, connecting {\L}ukasiewicz logic extended with multiplication with reals and neural networks of the above kind with real weights and biases via a representation theorem relating to McNaughton functions with real coefficients. Unlike the present paper, the activation values of the neural networks in the above characterisations, inputs and outputs included, were in the range $[0,1]$. The activation function was the truncated identity function, whereas the present paper uses the ordinary $\ReLU$-function. Moreover, the above articles did not consider \coolname or any analogous generalisation of neural networks.

There also exists a characterisation of neural networks via fuzzy logic that is not quite a straightforward logical characterisation. Castro and Trillas \cite{castro1998} gave a characterisation of neural networks that use real parameters and apply a squashing function on every layer via linear combinations of squashed transformations of propositions of {\L}ukasiewicz logic.

In contrast to fuzzy logic characterisations, non-fuzzy logical characterisations of neural networks trace back to McCulloch and Pitts \cite{mcculloch1943}, who characterised neural networks with Boolean activation values and threshold activation functions via temporal propositional expressions. Much more recently, Ahvonen, Heiman and Kuusisto \cite{DBLP:conf/csl/AhvonenHK24} gave a logical characterisation of recurrent neural networks operating with floating-point numbers via a recursive propositional logic called Boolean Network Logic (BNL). The same authors also gave an analogous characterisation of feedforward neural networks in \cite{ahvonen2026}. Whereas the activation values, inputs and outputs included, are Boolean in \cite{mcculloch1943} and floating-point numbers in \cite{DBLP:conf/csl/AhvonenHK24, ahvonen2026}, and whereas the truth values in the characterising logics are Boolean, in the current paper all activation values are real numbers, and likewise the truth values of the characterising fuzzy logics are reals from the closed interval $[0,1]$.

Regarding connections between $\LPH$ and real algebra, Marchioni and Montagna \cite{MarchioniM07} gave a constructive translation from terms of the language of ordered fields to terms of $\LPH$-algebra. This bears a close similarity to Lemma~\ref{lemma: terms to logic} of the present paper, where we give a constructive translation from terms of \coolname to formulae of the logic $\LPHF$. As in the current paper, the correspondence between terms and their translations was shown by the semantics of the field of reals and the canonical semantics of $\LPH$-algebras over the interval $[0,1]$. The translation in \cite{MarchioniM07} is uniform, based on a single non-linear function that bijectively maps the open interval $]0,1[$ to $\R$, and it was used to further interpret the universal theory of the reals in the equational theory of $\LPH$-algebras. By contrast, the translation in the current paper is based on a linear function, but it is non-uniform; for each positive integer $k$, we use a separate linear function that maps the interval $[-k,k]$ to the interval $[0,1]$. 
In Lemma~\ref{lemma: logic to terms}, we provide a reverse translation from $\LPHF$ to \coolname, but no reverse translation from terms of $\LPH$-algebra to terms of the language of ordered fields was given in \cite{MarchioniM07}.
Also, \coolname is more expressive than the language of ordered fields as it also includes rational constants and the function $\ReLU$, and likewise, our logic $\LPHF$ is less expressive than the logic $\LPH$ or the corresponding $\LPH$-algebra used in \cite{MarchioniM07}, i.e., we can express a wider variety of terms in a weaker logic.

Recently, there have been many logical characterisations of graph neural networks (GNNs). In \cite{barcelo2020}, Barceló et al. characterised GNNs in restriction to first-order logic (FO) via graded modal logic. They also showed that every property expressible by a formula of two-variable FO with counting is expressible by a GNN with global readout. Recently in \cite{hauke2025}, Hauke and Wa{\l}{\k{e}}ga showed that the reverse does not hold. 
In the wake of \cite{barcelo2020}, there have been numerous other (non-fuzzy) logical characterisations of GNNs \cite{grohe2024, benedikt2024, nunn2024, grau2025, schonherr2025, walega2026, luo2026} and of recurrent GNNs \cite{pflueger2024, ahvonen2024, selin2025}, and the lists sampled here are still growing.

There are various studies extending fuzzy logics with truth constants that generalise the ordinary truth constants $\top$ (``always true'') and $\bot$ (``always false'') used in Boolean logics. While the {\L}ukasiewicz implication $\varphi \impl \psi$ is valid exactly when the truth value of $\psi$ is always at least as great as the truth value of $\varphi$, the inclusion of truth constants $\overline{r}$ for all $r \in S$ for some $S \subseteq [0,1]$ allows the truth degrees of formulae to be studied independently. In particular, it allows formulae $\overline{r} \impl \psi$ which are valid exactly when the truth value of $\psi$ is always at least as great as $r$. This allows generalising the notions of truth and provability of formulae to cover various degrees of truth and provability. Pavelka-style completeness is when these generalised notions of truth and provability coincide, named after Jan Pavelka who showed in \cite{Pavelka79a, Pavelka79b, Pavelka79c} that {\L}ukasiewicz logic extended with a truth constant for each real in $[0,1]$, called Pavelka Logic, enjoys this property. While the number of truth constants makes the language of Pavelka Logic nondenumerable, Hájek showed in 
\cite{hajek1995fuzzy}
that Rational Pavelka Logic, obtained from Pavelka Logic by including truth constants only for rational numbers in $[0,1]$, also enjoys Pavelka-style completeness. In \cite{Hajek98}, Hájek also showed Pavelka-style completeness for the logic $\RPLM$ obtained from Rational Pavelka Logic by adding the product conjunction. Extensions of other fuzzy logics with truth constants have been studied, e.g., \cite{esteva2009, savicky2006, esteva2000, esteva2001, horvcik2004}. So far, truth constants do not appear to have a status other than as a technical means to calibrate degrees of truth., i.e., they do not appear to hold any independent status. By contrast, truth constants play a central role in the present paper, and our results establish a natural connection between truth constants and the weights and biases of neural networks.

\section{Preliminaries}\label{section: Preliminaries}

Let $\N$ denote the set of non-negative integers, $\Z$ the set of integers, $\Q$ the set of rational numbers and $\R$ the set of real numbers. We let $\N_+$, $\Z_+$, $\Q_+$ and $\R_+$ denote the restrictions of $\N$, $\Z$, $\Q$ and $\R$ to positive numbers. Given $a,b \in \R$ such that $a \leq b$, we let $[a,b]$ denote the closed interval $\{\, x \in \R \mid a \leq x \leq b \,\}$. Likewise, we define $]a,b[\,\, := [a,b] \setminus \{a,b\}$
and $[0,\infty[\,\, := \{\, x \in \R \mid x \geq 0 \,\}$.

Let $\PROP = \{\,p_i \mid i \in \N\,\}$ be a countably infinite set of proposition symbols $p_0, p_1, \dots$ and 
let $\VAR = \{\,x_i \mid i \in \N\,\}$ be a countably infinite set of variables $x_0, x_1, \dots$. We may use metasymbols $p, q, r, \dots$ to denote proposition symbols in $\PROP$ and $x, y, z, \dots$ to denote variables in $\VAR$. 
We fix a bijection $b : \PROP \to \VAR$, and for each $p \in \PROP$ we let $x_p$ denote $b(p)$, and for each $x \in \VAR$ we let $p_x$ denote $b^{-1}(x)$.

A \defstyle{valuation} is a function $\lmodel \colon \PROP \to [0,1]$ assigning a value $\lmodel(p) \in [0,1]$ to each proposition symbol $p \in \PROP$. 
Intuitively, $\lmodel(p)$ is the ``degree of truth'' of the proposition symbol $p$.
An \defstyle{evaluation map} is a function $\tmodel \colon \VAR \to \R$ that assigns a real value $\tmodel(x) \in \R$ to each variable $x \in \VAR$.

\subsection{Neural networks and $\ReLU$-polynomials}

In this section, we introduce neural networks and neurons as well as the generalised polynomial ring \coolname that subsumes neurons.

\subsubsection{\coolname}\label{subsection: coolname}

We begin by defining a generalised ring structure that our neural networks will be embedded in.
It is based on the polynomial ring $\Q[x_i]_{i \in \N}$, i.e., the polynomial ring over $\Q$ with countably many variables. We do not require these polynomials to be in a reduced form; for example, we consider both $x(y+z)$ and $xy + xz$ to be valid polynomial expressions. We only make one modification to $\Q[x_i]_{i \in \N}$: we add the function $\ReLU : \R \to [0,\infty[$, $\ReLU(x) = \max\{0,x\}$ as another building block.
We call such modified polynomials \emph{$\ReLU$-polynomials} or \emph{terms}, and the structure consisting of them as \coolname.

More formally, \defstyle{$\ReLU$-polynomials} or \defstyle{terms of \coolname} are constructed recursively as follows: 
\begin{itemize}
    \item For each $r \in \Q$, $\overline{r}$ is a term.
    \item Each $x \in \VAR$ is a term.
    \item If $t$ and $t'$ are terms, then $(t + t')$ and $(t \cdot t')$ are terms.
    \item If $t$ is a term, then $\ReLU(t)$ is a term.
\end{itemize}
When there is no danger of confusion, we omit the parentheses from $(t + t')$ and $(t \cdot t')$.
For brevity, we will usually write $tt'$ to denote $t \cdot t'$. We will also write $-t$ to denote $\overline{-1} \cdot t$ and $t - t'$ to denote $t + (-t')$.

We interpret terms of \coolname over the field of real numbers as follows.
An evaluation map $\tmodel$
fixes a value for each term when $+$, $\cdot$ and $\ReLU$ are given natural interpretations over the field of reals.
\begin{itemize}
    \item If $t = \overline{r}$ for some $r \in \Q$, then $\tmodel(t) = r$.
    \item If $t = t' + t''$ for some terms $t'$ and $t''$, then $\tmodel(t) = \tmodel(t') + \tmodel(t'')$.\footnote{Here and throughout the paper, we abuse notation and use $+$ and $\cdot$ to denote both a symbol in the syntax of \coolname and the standard sum and multiplication over the field of real numbers. We could use different symbols like $\boxplus$ and $\boxdot$ in the syntax to differentiate the syntax and the semantics, but this would be more likely to increase confusion than decrease it. In the case of rationals however, it is useful to make the difference between syntax and semantics explicit, because a rational number can be represented in infinitely many ways, all of which can be covered with a single constant symbol.}
    \item If $t = t' \cdot t''$ for some terms $t'$ and $t''$, then $\tmodel(t) = \tmodel(t')\tmodel(t'')$.
    \item If $t = \ReLU(t')$ for some term $t'$, then $\tmodel(t) = \max\{0, \tmodel(t')\}$.
\end{itemize}
For each $k \in \Z_+$ and terms $t_1, \dots, t_k$, we write $\sum_{i = 1}^{k}t_i$ to denote the term $t_1 + \cdots + t_n$. The omitted parentheses in the sum can be placed in some canonical way, but given the semantics above, $+$ is associative and thus it does not make a difference how the parentheses are placed.

\begin{example}
    Let $t := \ReLU(x + xy) + y$. Now $t$ is a term of \coolname, and given an evaluation map $\tmodel$ such that $\tmodel(x) = 3$ and $\tmodel(y) = 0.5$, we have 
    \[
        \tmodel(t) = \max\{0,3 + 3 \cdot 0.5\} + 0.5 = \max\{0,4.5\} + 0.5 = 4.5 + 0.5 = 5.
    \]
\end{example}

Next, we define 
the degree of terms.
Intuitively, degree relates to the maximum number of times that a variable is multiplied with another variable.

More formally, the \defstyle{degree} of a term $t$, denoted $\deg(t)$, is defined recursively as follows:
\begin{itemize}
    \item If $t = \overline{r}$ for some $r \in \Q$, then $\deg(t) = 0$.
    \item If $t = x$ for some $x \in \VAR$, then $\deg(t) = 1$.
    \item If $t = t' + t''$ for some terms $t'$ and $t''$, then $\deg(t) = \max\{\deg(t'), \deg(t'')\}$.
    \item If $t = t' t''$ for some terms $t'$ and $t''$, then $\deg(t) = \deg(t') + \deg(t'')$.
    \item If $t = \ReLU(t')$ for some term $t'$, then $\deg(t) = \deg(t')$.
\end{itemize}
When we discuss \defstyle{\coolname with degree bound $d \in \N$}, we mean the collection of exactly those terms $t$ of \coolname for which $\deg(t) \leq d$.

\subsubsection{Neural networks}\label{section: Neural networks}

Next, we discuss neural networks. We begin by defining an auxiliary notion of a proto-neuron, which we will use throughout the paper. Then we define neural networks and neurons. Proto-neurons resemble neurons, but offer a more flexible syntax to work with, which is useful when constructing translations.

First, we consider proto-neurons, which are intuitively terms of \coolname of degree at most $1$. Syntactically, \defstyle{proto-neurons} are constructed recursively as follows:
\begin{itemize}
    \item For each $r \in \Q$, $\overline{r}$ is a proto-neuron.
    \item Each $x \in \VAR$ is a proto-neuron.
    \item If $t$ and $t'$ are proto-neurons, then $t+t'$ is a proto-neuron.
    \item If $t$ and $t'$ are proto-neurons such that $\deg(t) + \deg(t') \leq 1$, then $tt'$ is a proto-neuron.
    \item If $t$ is a proto-neuron, then $\ReLU(t)$ is a proto-neuron.
\end{itemize}

\begin{restatable}{lemma}{degreelemma}\label{lemma: proto-neuron degrees}
    For all terms $t$ of \coolname, $t$ is a proto-neuron if and only if $\deg(t) \leq 1$.
\end{restatable}
\begin{proof}
    We prove the claim by induction over the structure of $t$; the details are in Appendix~\ref{appendix: neural networks}.
\end{proof}

\begin{example}
    The term $3(x+y)$ is a proto-neuron, but the term $(x+3)y$ is not. In general, $t$ is a proto-neuron if 
    it does not involve taking the product of two terms that both contain variables.
\end{example}

Proto-neurons, as the name suggests, resemble the neurons of a neural network. However, neurons only constitute a subclass of proto-neurons. Next, we introduce the concept of a neuron, which is closely related to the architecture of neural networks, which we shall define simultaneously. For an intuitive description of neural networks and neurons, see Section~\ref{section: introduction}.

A \defstyle{rational-parameter $\ReLU$-activated neural network of depth $d \in \N$} is defined recursively as follows. A neural network of depth $0$ is a tuple $(y_1, \dots, y_k)$ where $k \in \Z_+$ and $y_1, \dots, y_k$ are distinct variables in $\VAR$. Next, assume we have defined neural networks of depth $d$. A neural network of depth $d+1$ is a tuple $N = (n_1, \dots, n_k)$ where $k \in \Z_+$ and each $n_j$ is a proto-neuron of the form $\ReLU\left(\sum_{i=1}^{\ell} \overline{w_{i,j}} n_i' + \overline{b_j} \right)$ where $\ell \in \Z_+$, $w_{i,j}, b_j \in \Q$ and $(n_1', \dots, n_{\ell}')$ is a neural network of depth~$d$. Note that $\ell$ and the neurons $n_1', \dots, n_{\ell}'$ are the same regardless of $j$. The components $n_1, \dots, n_k$ of a neural network are called \defstyle{neurons}; in particular, they are the \defstyle{output neurons of $N$}. Each $\overline{w_{i,j}}$ is called the \defstyle{weight} from the neuron $n_i'$ to the neuron $n_j$,\footnote{Technically, there may be multiple such weights if some of the neurons $n_1', \dots, n_{\ell}'$ are identical to each other, since our definition of a neuron simply refers to an aggregation formula, but in such situations we naturally consider the identical neurons to be distinct from each other since they occupy separate positions in the neural network.} and each $\overline{b_j}$ is called the \defstyle{bias} of the neuron $n_j$. 

\begin{example}
    The proto-neuron $-x$ is not a neuron and neither is $\ReLU(\overline{2}x) + \ReLU(\overline{3}y)$, but the following 
    proto-neuron is a neuron:
    \[
        \ReLU\left(\overline{1} \cdot \ReLU(\overline{2}x + \overline{0}y + \overline{0}) + \overline{1} \cdot \ReLU(\overline{0}x + \overline{3}y + \overline{0}) + \overline{0}\right).
    \]
\end{example}

\subsection{Logics}\label{section: logics}
 
In this section, we introduce the variety of fuzzy logics that will be used to characterise \coolname and neural networks. Key roles here are played by Rational Pavelka Logic $\RPL$ as well as another fuzzy logic called $\LPH$. Both of these are extensions of {\L}ukasiewicz logic, which we will define first.
 
\subsubsection{{\L}ukasiewicz logic}

The first logic we define is {\L}ukasiewicz logic, which is syntactically analogous to propositional logic but interpreted differently. While the canonical semantics for propositional logic gives each formula a truth value $0$ (false) or $1$ (true), the canonical semantics for infinite-valued {\L}ukasiewicz logic generalises this, giving each formula a truth value from the interval $[0,1]$ intuitively allowing degrees of truth between true and false. Under this interpretation, certain arithmetic operations like sum and subtraction are definable, though truncated at the end points $0$ and $1$ of the interval.

The \defstyle{$\PROP$-formulae of infinite-valued {\L}ukasiewicz logic} \cite{lukasiewicz1930} (also called {\L}ukasiewicz logic or the {\L}ukasiewicz-Tarski logic) are defined by the following grammar:
\[
    \varphi ::= \overline{0} \,|\, p \,|\, \varphi \impl \varphi,
\]
where $p \in \PROP$. Intuitively, $\overline{0}$ is always false, and $\varphi \impl \psi$ is read ``if $\varphi$, then $\psi$''. Note that similar to terms of \coolname, logic formulae are also terms, only of a different language, but we call them formulae to keep the terminologies distinct from each other.

The semantics of infinite-valued {\L}ukasiewicz logic is defined over MV-algebras \cite{chang1958, mundici2011}. Rather than presenting the full variety of MV-algebras, we present the semantics given by the canonical MV-algebra defined over the real interval $[0,1]$ which satisfies all the axioms of {\L}ukasiewicz logic.
Let $\lmodel : \PROP \to [0,1]$ be a valuation. We extend the function~$\lmodel$ to determine the truth value of any formula of {\L}ukasiewicz logic as follows:
\begin{itemize}
    \item $\lmodel(\overline{0}) := 0$.
    \item $\lmodel(\varphi \impl \psi) := \min\{1, 1 - \lmodel(\varphi) + \lmodel(\psi)\}$ for all formulae $\varphi$ and $\psi$ of {\L}ukasiewicz logic.
\end{itemize}
For the reader more familiar with propositional logic, it may be useful to check how the connective $\impl$ and subsequent connectives behave when $\lmodel(\varphi), \lmodel(\psi) \in \{0,1\}$. In the case of $\impl$, such an examination retrieves the canonical Boolean semantics for implication.

We define the following abbreviated connectives with the intuitive meaning given in parentheses: $\negl$~(negation), $\oplus$~(addition), $\otimes$~(addition with bias $-1$), $\ominus$~(subtraction), $\land$~(minimum), $\lor$~(maximum) and $\equivl$~(equivalence). The connectives $\negl$, $\oplus$ and $\ominus$ will be crucial later on, whereas the others may be of independent interest to the reader. We give a definition for each connective on the left and the resulting truth value of the formula in a valuation $\lmodel$ on the right (the details are left for the reader):
\[
\begin{aligned}
    &\negl \varphi := \varphi \impl \overline{0}, & &\lmodel(\negl \varphi) = 1 - \lmodel(\varphi), \\
    &\varphi \oplus \psi := \negl \varphi \impl \psi, & &\lmodel(\varphi \oplus \psi) = \min\{1, \lmodel(\varphi) + \lmodel(\psi)\}, \\
    &\varphi \pand \psi := \negl(\varphi \impl \negl\psi), & &\lmodel(\varphi \pand \psi) = \max\{0, \lmodel(\varphi) + \lmodel(\psi) - 1\}, \\
    &\varphi \ominus \psi := \varphi \pand \negl \psi, & &\lmodel(\varphi \ominus \psi) = \max\{0, \lmodel(\varphi) - \lmodel(\psi)\}, \\
    &\varphi \land \psi := \varphi \pand (\varphi \impl \psi), & &\lmodel(\varphi \land \psi) = \min\{\lmodel(\varphi), \lmodel(\psi)\}, \\
    &\varphi \lor \psi := (\varphi \impl \psi) \impl \psi, & &\lmodel(\varphi \lor \psi) = \max\{\lmodel(\varphi), \lmodel(\psi)\}, \\
    &\varphi \equivl \psi := (\varphi \impl \psi) \land (\psi \impl \varphi), & &\lmodel(\varphi \equivl \psi) = 1 - |\lmodel(\varphi) - \lmodel(\psi)|.
\end{aligned}
\]
These connectives are naturally definable in the same way in all extensions of {\L}ukasiewicz logic.

\subsubsection{Rational Pavelka Logic}

Next, we define Rational Pavelka Logic ($\RPL$), which extends {\L}ukasiewicz logic by adding a truth constant for each rational number in the interval $[0,1]$. Pavelka logic was originally introduced in \cite{Pavelka79a, Pavelka79b, Pavelka79c}, where a truth constant was included for each real number in the interval $[0,1]$ including for irrational numbers. Later in \cite{hajek1995fuzzy} the rational version of Pavelka Logic was introduced, where truth constants are only included for rational numbers.

The \defstyle{$\PROP$-formulae of Rational Pavelka Logic ($\RPL$)} are defined by the following grammar:
\[
    \varphi ::= \overline{r} \,|\, p \,|\, \varphi \impl \varphi,
\]
where $r \in \Q \cap [0,1]$ and $p \in \PROP$. Intuitively, $\overline{r}$ is a formula that always has the truth value $r$.

The canonical semantics of $\RPL$ over the real interval $[0,1]$ extends the canonical semantics of {\L}ukasiewicz logic as follows. Let $\lmodel : \PROP \to [0,1]$ be a valuation. The semantics for proposition symbols and $\impl$ is defined as for {\L}ukasiewicz logic. For truth constants, we define
\begin{itemize}
    \item $\lmodel(\overline{r}) := r$ for each $r \in \Q \cap [0,1]$.
\end{itemize}

\subsubsection{$\RPLM$}

In \cite{Hajek98}, Hájek introduced the logic $\RPLM$, which extends $\RPL$ with the product conjunction $\odot$.\footnote{In the literature on MV-algebras, the connective $\odot$ commonly refers to the dual operator of $\oplus$ denoted in this paper by $\otimes$. In this work, $\odot$ denotes a primitive connective that is not derived from $\oplus$ or $\impl$.} As the name suggests, the product conjunction allows multiplying the truth value of a formula with the truth value of another formula.

The \defstyle{$\PROP$-formulae of $\RPLM$} are defined by the following grammar:
\[
    \varphi ::= \overline{r} \,|\, p \,|\, \varphi \impl \varphi \,|\, \varphi \odot \varphi,
\]
where $r \in \Q \cap [0,1]$ and $p \in \PROP$. Intuitively, $\odot$ corresponds to multiplication.

The canonical semantics of $\RPLM$ over the real interval $[0,1]$ extends the semantics of $\RPL$ as follows. Let $\lmodel : \PROP \to [0,1]$ be a valuation. 
The semantics for truth constants, proposition symbols and $\impl$ is defined as in $\RPL$. For $\odot$, we define
\begin{itemize}
    \item $\lmodel(\varphi \odot \psi) := \lmodel(\varphi) \lmodel(\psi)$ for all formulae $\varphi$ and $\psi$ of $\RPLM$.
\end{itemize}

\subsubsection{$\LPH$}\label{section: LPH}

The logic $\LPH$ is another extension of {\L}ukasiewicz logic. Rather than adding a truth constant for each rational number in $[0,1]$ like $\RPL$ and $\RPLM$, $\LPH$ only adds one for~$\half$. 
On top of the product conjunction $\odot$, $\LPH$ adds the corresponding residual division connective: the product implication $\impp$.

The \defstyle{$\PROP$-formulae of $\LPH$} are obtained by the following grammar:
\[
    \textstyle{\varphi ::= \overline{0} \,|\, \overline{\half} \,|\, p \,|\, \varphi \impl \varphi \,|\, \varphi \odot \varphi \,|\, \varphi \impp \varphi,}
\]
where $p \in \PROP$. Intuitively, $\impp$ corresponds to division.

The canonical semantics of $\LPH$ over $[0,1]$ is defined as follows. Let $\lmodel : \PROP \to [0,1]$ be a valuation.
The semantics for truth constants, proposition symbols, $\impl$ and $\odot$ is defined as in $\RPLM$. The semantics of $\impp$ is defined as follows. Given formulae $\varphi$ and $\psi$ of $\LPH$, 
\[
\lmodel(\varphi \impp \psi) := \begin{cases}
    1, &\text{ if } \lmodel(\varphi) \leq \lmodel(\psi) \\
    \frac{\lmodel(\psi)}{\lmodel(\varphi)}, &\text{otherwise}.
\end{cases}
\]
Note that $\lmodel(\varphi) = 0$ implies $\lmodel(\varphi) \leq \lmodel(\psi)$. Thus, the denominator of $\frac{\lmodel(\psi)}{\lmodel(\varphi)}$ never goes to zero.

It is shown in \cite{esteva2001} that for each $r \in \Q \cap [0,1]$, we can construct a formula $\underline{r}$ of $\LPH$ without proposition symbols such that $\lmodel(\underline{r}) = r$ for each valuation $\lmodel$. We show an alternative construction in Appendix~\ref{appendix: defining rationals}.

\subsubsection{Fragments of $\RPLM$ and $\LPH$}\label{section: fragments}

In this section, we define various syntactic fragments of $\LPH$ and $\RPLM$. We first establish two fragments of $\LPH$ that place limitations on the use of the connectives $\odot$ and $\impp$. Then, we define a hierarchy of fragments of $\LPH$ and $\RPLM$, starting with fragments containing no proposition symbols and fragment-by-fragment allowing more proposition symbols nested within product conjunctions $\odot$.

We start with an auxiliary fragment of $\LPH$ that consists of the formulae of $\LPH$ that do not contain proposition symbols. This means that the truth values of formulae do not depend on the valuation. More formally, the \defstyle{formulae of $\LPH_{\leq 0}$} are constructed according to the following grammar:
\[
    \textstyle{\varphi ::= \overline{0} \,|\, \overline{\half} \,|\, \varphi \impl \varphi \,|\, \varphi \odot \varphi \,|\, \varphi \impp \varphi.}
\]

Next, we define a fragment of $\LPH$ where proposition symbols do not appear in the scope of the connectives $\odot$ and $\impp$. This means the truth values of multiplications and divisions do not depend on the valuation. The \defstyle{$\PROP$-formulae of $\LPHFF$} are constructed according to the grammar:
\[
    \textstyle{\varphi ::= \psi \,|\, p \,|\, \varphi \impl \varphi,}
\]
where $\psi$ is a formula of $\LPH_{\leq 0}$ and $p \in \PROP$.

Analogously, we define a fragment of $\LPH$ that only restricts the use of $\impp$, i.e., only the truth values of divisions do not depend on the valuation. The \defstyle{$\PROP$-formulae of $\LPHF$} are constructed according to the following grammar:
\[
    \varphi ::= \psi \,|\, p \,|\, \varphi \impl \varphi \,|\, \varphi \odot \varphi,
\]
where $\psi$ is a formula of $\LPH_{\leq 0}$ and $p \in \PROP$.

Next, we define a hierarchy of syntactic fragments of both $\RPLM$ and $\LPHF$.
First, we define the fragments $\RPLM_{\leq 0}$ and $\LPHF_{\leq 0}$ which do not contain proposition symbols. In the case of $\LPHF$, the \defstyle{formulae of $\LPHF_{\leq 0}$} are simply those of the logic $\LPH_{\leq 0}$ defined above. The \defstyle{formulae of $\RPLM_{\leq 0}$} are obtained according to the following grammar:
\[
    \varphi ::= \overline{r} \,|\, \varphi \impl \varphi \,|\, \varphi \odot \varphi,
\]
where $r \in \Q \cap [0,1]$.

Next, let $\cL$ be either $\RPLM$ or $\LPHF$ and assume we have defined the fragment $\cL_{\leq n}$ where $n \in \N$. The \defstyle{$\PROP$-formulae of $\cL_{\leq n+1}$} are obtained according to the following grammar:
\[
    \varphi ::= p \,|\, \psi \,|\, \alpha \odot \beta \,|\, \varphi \impl \varphi,
\]
where $p \in \PROP$, $\psi$ is a formula of $\cL_{\leq n}$, and $\alpha$ and $\beta$ are formulae of $\cL_{\leq i}$ and $\cL_{\leq j}$ respectively for any $i,j \in \N$ such that $i + j \leq n + 1$. If $n \geq 1$, then the atomic proposition symbols $p$ could be omitted from the grammar, as they would be implicitly included with~$\psi$.
The fragments $\cL_{\leq 0}$, $\cL_{\leq 1}$, $\dots$ form a hierarchy where each fragment is subsumed by the next.

\begin{lemma}\label{lemma:degrees of formulae}
    Let $\cL$ be $\RPLM$ or $\LPHF$. Let $i,j \in \N$, let $\varphi$ be a formula of $\cL_{\leq i}$ and let $\psi$ be a formula of $\cL_{\leq j}$. Then the following hold.
    \begin{itemize}
        \item For each $r \in \Q \cap [0,1]$, $\overline{r}$ is a formula of $\RPLM_{\leq 0}$ and $\underline{r}$ is a formula of $\LPHF_{\leq 0}$.
        \item Each $p \in \PROP$ is a formula of $\cL_{\leq 1}$.
        \item $\negl \varphi$ is a formula of $\cL_{\leq i}$.
        \item $\varphi \impl \psi$, $\varphi \oplus \psi$, $\varphi \otimes \psi$ and $\varphi \ominus \psi$ are formulae of $\cL_{\leq \max\{i, j\}}$.
        \item $\varphi \odot \psi$ is a formula of $\cL_{\leq i + j}$.
    \end{itemize}
\end{lemma}
\begin{proof}
    The cases for $\overline{r}$, $p$, $\varphi \impl \psi$ and $\varphi \odot \psi$ follow from the definition of the fragments $\cL_{\leq n}$. The formula $\underline{r}$ is built without proposition symbols and thus a formula of $\LPHF_{\leq 0}$. The cases for $\negl \varphi$, $\varphi \oplus \psi$, $\varphi \otimes \psi$ and $\varphi \ominus \psi$ are then easy to check by the definitions of the abbreviated connectives $\negl$, $\oplus$, $\otimes$ and $\ominus$, which are constructed using only the truth constant $\overline{0}$ and the implication $\impl$.
\end{proof}

\section{Equivalence}\label{section: Equivalence}

In this section, we define various notions of equivalence between terms of \coolname and/or formulae of logics. 
The most basic notions are centred around formulae and/or terms receiving identical values with matching valuations and evaluation maps. A more involved notion of equivalence involves scaling the values of real intervals down to the interval $[0,1]$, since unlike terms of \coolname, the truth values of formulae of $\RPLM$ and $\LPH$ are always in $[0,1]$. 

\subsection{Non-scaled equivalence}

In this section, we consider notions of equivalence based on receiving identical values with identical inputs. We start by defining equivalence between terms.

\begin{definition}
    We say that terms $t$ and $t'$ of \coolname are \defstyle{equivalent} if we have $\tmodel(t) = \tmodel(t')$ for each evaluation map $\tmodel$.
    For two tuples of terms $T = (t_1, \dots, t_k)$ and $T' = (t_1', \dots, t_k')$ of equal length, we say that $T$ and $T'$ are \defstyle{equivalent} if the terms $t_i$ and~$t_i'$ are equivalent for each $i \in \{1, \dots, k\}$.
\end{definition}

\begin{example}
    The terms $x+y$ and $\overline{\half} x + \overline{\half} x + \overline{2}y - y$ are equivalent. The terms $x$ and $\ReLU(x)$ are not equivalent: if $\tmodel$ is an evaluation map such that $\tmodel(x) < 0$, then we get $\tmodel(\ReLU(x)) = \max\{0,\tmodel(x)\} = 0 \neq \tmodel(x)$.
\end{example}

Next, we establish equivalence between formulae of logics.

\begin{definition}
    We say that two formulae $\varphi$ and $\psi$ of $\RPLM$ or $\LPH$ are \defstyle{equivalent} if $\lmodel(\varphi) = \lmodel(\psi)$ for each valuation~$\lmodel$.
\end{definition}

\begin{example}
    The formulae $\overline{\frac{1}{4}}$ and $\overline{\half} \odot \overline{\half}$ are equivalent. So are the formulae $\varphi$ and $\negl\negl\varphi$.
\end{example}

Lastly, we consider normal equivalence between terms of \coolname and formulae of logics.

\begin{definition}
    We say that a term $t$ of \coolname and a formula $\varphi$ of $\RPLM$ or $\LPH$ are \defstyle{equivalent} if 
    given any evaluation map $\tmodel$ and valuation $\lmodel$ such that $\tmodel(x) = \lmodel(p_x)$ for all $x \in \VAR$: 
    \[
        \tmodel(t) = \lmodel(\varphi).
    \]
\end{definition}

\subsection{Connecting fragments of $\RPLM$ and $\LPH$}

In this section, we establish a non-scaled notion of expressive power between logics and use it to match fragments of $\RPLM$ and $\LPH$.

\begin{definition}
    We say that two logics $\cL$ and $\cL'$ \defstyle{have the same expressive power} if for each formula of $\cL$ there is an equivalent formula of $\cL'$ and vice versa.
\end{definition}

We next establish natural connections between fragments of $\RPLM$ and $\LPH$.

\begin{proposition}\label{proposition: RPLM_d = LPHF_d}
    For all $n \in \N$, the logics $\RPLM_{\leq n}$ and $\LPHF_{\leq n}$ have the same expressive~power.
\end{proposition}
\begin{proof}
    We show the claim by induction over $n$. The interesting part is the base case where $n = 0$, as the induction step is then trivial by the identical way in which the fragments $\RPLM_{\leq n}$ and $\LPHF_{\leq n}$ are constructed from $\RPLM_{\leq 0}$ and $\LPHF_{\leq 0}$, respectively. First, note that for each formula $\varphi$ of either $\RPLM_{\leq 0}$ or $\LPHF_{\leq 0}$, we have $\lmodel(\varphi) = \lmodel'(\varphi) = r \in \Q \cap [0,1]$ for all valuations $\lmodel$ and $\lmodel'$; this is easy to see because $\varphi$ contains no proposition symbols, which guarantees that the truth value of $\varphi$ is a constant and indeed a rational number. Now, by the semantics of $\underline{r}$ and Lemma~\ref{lemma:degrees of formulae}, if $\varphi$ is a formula of $\RPLM_{\leq 0}$, then $\underline{r}$ is an equivalent formula of $\LPHF_{\leq 0}$. Likewise, if $\varphi$ is a formula of $\LPHF_{\leq 0}$, then $\overline{r}$ is an equivalent formula of $\RPLM_{\leq 0}$.
\end{proof}

By similar reasoning, we also obtain the following two propositions.

\begin{proposition}\label{proposition: RPLM = LPHF}
    $\RPLM$ and $\LPHF$ have the same expressive power.
\end{proposition}

\begin{proposition}\label{proposition: RPL = LPHFF}
    $\RPL$ and $\LPHFF$ have the same expressive power.
\end{proposition}

\subsection{Scaled equivalence}

In this section we define scaled notions of equivalence. First, we define a necessary scaling function.
\begin{definition}
    Let $k \in \R_+$. We define the bijection $\scale_k : \R \to \R$, $\scale_k(x) = \frac{k + x}{2k}$. 
\end{definition}
Clearly, $\scale_k$ scales the interval $[-k,k]$ to the interval $[0,1]$, as we have $\scale_k(-k) = 0$, $\scale_k(0) = \half$ and $\scale_k(k) = 1$. 
The idea behind the following notion of equivalence is that the truth values of formulae being restricted to $[0,1]$, a formula can simulate a term  
when both the inputs and outputs are scaled down using $\scale_k$. This naturally means that the term can only gain values within the interval $[-k,k]$, and to accomplish this, the input interval also naturally needs to be restricted, as otherwise a term such as $x$ could obtain arbitrarily great values.

\begin{definition}
    Let $i,k \in \R_+$ such that $i \leq k$. Let $\varphi$ be a formula of $\RPLM$ or $\LPH$ and let $t$ be a term of \coolname. We say that $\varphi$ is \defstyle{$(i,k)$-equivalent} to $t$ if 
    given any evaluation map $\tmodel$ such that $\tmodel(x) \in [-i,i]$ for each $x \in \VAR$ and valuation
    $\lmodel$ such that $\lmodel(p) = \scale_k(\tmodel(x_p))$ for all $p \in \PROP$, we have
    \[
        \lmodel(\varphi) = \scale_k(\tmodel(t)).
    \]
\end{definition}

Lastly, we use the scaling function $\scale_k$ to define an alternative notion of equivalence between formulae of logics.
This equivalence notion applies $\scale_k$ between two logic formulae. Intuitively, this means that a formula simulates another formula, but does so in the interval $\left[\half, \frac{k+1}{2k}\right]$.

    \begin{definition}
    Let $k \in \R_+$, let $\varphi$ be a formula of $\RPLM$ or $\LPH$ and likewise for~$\psi$. We say that $\psi$ is \defstyle{$k$-equivalent} to $\varphi$ if
    given any two valuations $\lmodel$ and $\lmodel'$ such that $\lmodel'(p) = \scale_k(\lmodel(p))$ for each $p \in \PROP$, we have
    \[
        \lmodel'(\psi) = \scale_k(\lmodel(\varphi)).
    \]
    \end{definition}

We establish a useful lemma, connecting the above notion of equivalence to scaled and non-scaled equivalence between terms and formulae.

\begin{lemma}\label{lemma: transitive equivalence}
    Let $i,k \in \R_+$ such that $1 \leq i \leq k$, let $\varphi$ and $\psi$ each be a formula of $\RPLM$ or $\LPH$ and let $t$ be a term of \coolname. If $\varphi$ and $t$ are equivalent and $\psi$ is $(i,k)$-equivalent to $t$, then $\psi$ is $k$-equivalent to $\varphi$.
\end{lemma}
\begin{proof}
    Let $\lmodel$ be a valuation and let $\lmodel'$ be the valuation such that $\lmodel'(p) = \scale_k(\lmodel(p))$ for each $p \in \PROP$. Let $\tmodel$ be the unique evaluation map such that $\tmodel(x) = \lmodel(p_x)$ for each $x \in \VAR$, which also implies $\tmodel(x) \in [0,1] \subset [-i,i]$ for each $x \in \VAR$ as well as
    $\lmodel'(p) = \scale_k(\tmodel(x_p))$ for each $p \in \PROP$. Because $\psi$ is $(i,k)$-equivalent to $t$, we have $\lmodel'(\psi) = \scale_k(\tmodel(t))$, and because $\varphi$ and $t$ are equivalent, we have $\tmodel(t) = \lmodel(\varphi)$. Thus $\lmodel'(\psi) = \scale_k(\lmodel(\varphi))$, i.e., $\psi$ is $k$-equivalent to $\varphi$.
\end{proof}

\subsection{Scaled expressive power}

In this section, we define scaled notions of expressive power between \coolname and/or logics. We start with a scaled notion of expressive power between logics.

\begin{definition}
Let $\cL$ and $\cL'$ be logics. We say that $\cL$ has \defstyle{the same expressive power as $\cL'$ (w.r.t. scaling)} if the following conditions hold:
\begin{enumerate}
    \item for each formula of $\cL$, there exists a $k$-equivalent formula of $\cL'$ for some $k \in \R_+$, and
    \item for each formula of $\cL'$, there exists an equivalent formula of $\cL$.
\end{enumerate}
\end{definition}

Lastly, we establish a scaled notion of expressive power that we use in our characterisation theorems.

\begin{definition}\label{definition: logical characterisation}
Let $\cT$ be a subclass of \coolname (or of tuples thereof) and let $\cL$ be one of the defined logics. We say that $\cT$ and $\cL$ have \defstyle{the same expressive power (w.r.t. scaling)} if for each $n \in \Z_+$ the following conditions hold:
\begin{enumerate}
    \item for each tuple $(t_1, \dots, t_n)$ of $\cT$ and all $i \in \R_+$, there exists some 
    $k \geq i$ 
    and a tuple $(\varphi_1, \dots, \varphi_n)$ of formulae of $\cL$ such that for each $j \in \{1, \dots, n\}$ the formula $\varphi_j$ is $(i,k)$-equivalent to the term $t_j$, and
    \item for each tuple $(\varphi_1, \dots, \varphi_n)$ of formulae of $\cL$ there exists a tuple $(t_1, \dots, t_n)$ of $\cT$ such that for each $j \in \{1, \dots, n\}$ the term $t_j$ and the formula $\varphi_j$ are equivalent.
\end{enumerate}
\end{definition}
Note that although $\cT$ can be either a subclass of terms or tuples of terms of \coolname, in both cases we translate tuples of terms to tuples of formulae and vice versa.
Appropriately, the numbers $i$ and $k$ in the first condition must be independent of which components of the tuple are being compared. For an illustration of how the scaling works in this kind of characterisation, see Figure~\ref{figure: scaling}.

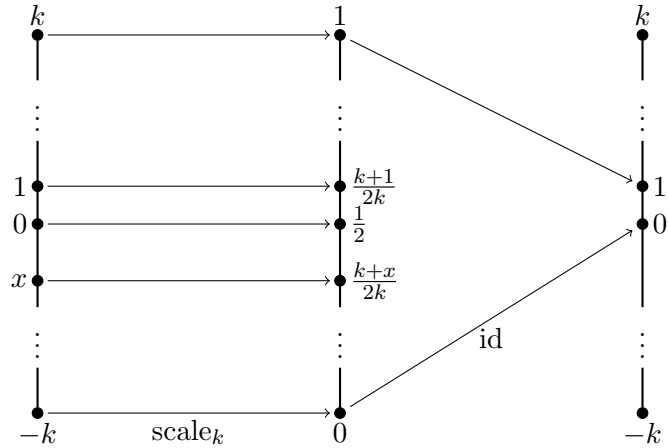
\begin{figure}[ht]
    \begin{center}
    \begin{tikzpicture}
        \node at (0,0.625) (inf) {};
        \node at (0,0.5) (k) {};
        \node at (0,-0.5) (2) {$\vdots$};
        \node at (0,-1.5) (1) {};
        \node at (0,-2) (0) {};
        \node at (0,-3.5) (-2) {$\vdots$};
        \node at (0,-2.75) (x) {};
        \node at (0,-4.5) (-k) {};
        \node at (0,-4.625) (-inf) {};

        \node at (2,-4.75) (scaling) {$\scale_k$};

        \draw[thick] (2) to (-2);
        \draw[thick] (inf) to (2);
        \draw[thick] (-2) to (-inf);

        \filldraw[black] (0,0.5) circle (2pt) node[anchor=south]{$k$};
        \filldraw[black] (0,-1.5) circle (2pt) node[anchor=east]{$1$};
        \filldraw[black] (0,-2) circle (2pt) node[anchor=east]{$0$};
        \filldraw[black] (0,-2.75) circle (2pt) node[anchor=east]{$x$};
        \filldraw[black] (0,-4.5) circle (2pt) node[anchor=north]{$-k$};

        \node at (4,0.625) (inf') {};
        \node at (4,0.5) (k') {};
        \node at (4,-0.5) (2') {$\vdots$};
        \node at (4,-1.5) (1') {};
        \node at (4,-2) (0') {};
        \node at (4,-3.5) (-2') {$\vdots$};
        \node at (4,-2.75) (x') {};
        \node at (4,-4.5) (-k') {};
        \node at (4,-4.625) (-inf') {};

        \node at (6,-3.5) (id) {$\id$};

        \draw[thick] (2') to (-2');
        \draw[thick] (inf') to (2');
        \draw[thick] (-2') to (-inf');

        \filldraw[black] (4,0.5) circle (2pt) node[anchor=south]{$1$};
        \filldraw[black] (4,-1.5) circle (2pt) node[anchor=west]{$\frac{k+1}{2k}$};
        \filldraw[black] (4,-2) circle (2pt) node[anchor=west]{$\half$};
        \filldraw[black] (4,-2.75) circle (2pt) node[anchor=west]{$\frac{k+x}{2k}$};
        \filldraw[black] (4,-4.5) circle (2pt) node[anchor=north]{$0$};

        \draw[->] (k) to (k');
        \draw[->] (1) to (1');
        \draw[->] (0) to (0');
        \draw[->] (x) to (x');
        \draw[->] (-k) to (-k');

        \node at (8,0.625) (inf'') {};
        \node at (8,0.5) (k'') {};
        \node at (8,-0.5) (2'') {$\vdots$};
        \node at (8,-1.5) (1'') {};
        \node at (8,-2) (0'') {};
        \node at (8,-3.5) (-2'') {$\vdots$};
        \node at (8,-4.5) (-k'') {};
        \node at (8,-4.625) (-inf'') {};

        \draw[thick] (2'') to (-2'');
        \draw[thick] (inf'') to (2'');
        \draw[thick] (-2'') to (-inf'');

        \filldraw[black] (8,0.5) circle (2pt) node[anchor=south]{$k$};
        \filldraw[black] (8,-1.5) circle (2pt) node[anchor=west]{$1$};
        \filldraw[black] (8,-2) circle (2pt) node[anchor=west]{$0$};
        \filldraw[black] (8,-4.5) circle (2pt) node[anchor=north]{$-k$};
        
        \draw[->] (k') to (1'');
        \draw[->] (-k') to (0'');
    \end{tikzpicture}
    \end{center}
    \caption{An illustration of the scaling in our translations. 
    The left and right number lines correspond to values of terms, while the middle number line shows the truth values of formulae. 
    From terms to formulae (left-to-middle), we scale $[-k,k]$ linearly into $[0,1]$, while from formulae to terms (middle-to-right) we use the identity function.}
    \label{figure: scaling}
\end{figure}

\section{Fuzzy logic characterisations}

In this section, we give fuzzy logic characterisations of \coolname and specifically of neural networks. 
In Section~\ref{section: Characterising coolname}, we give a fuzzy logic characterisation of \coolname, as well as of subclasses of \coolname obtained by giving a constant upper bound for the degrees of terms. Then, we give a fuzzy logic characterisation of neural networks in Section~\ref{section: Characterising neural networks}.

\subsection{Characterising \coolname}\label{section: Characterising coolname}

In this section, we give a fuzzy logic characterisation for \coolname with and without degree bound via the logics $\RPLM$ and $\LPHF$. As a reminder, in our characterisations we translate tuples of terms to tuples of formulae and vice versa.

\firsttheorem
\begin{proof}
    (Sketch) First, we define auxiliary connectives in $\RPLM$ and $\LPHF$ that simulate the addition and multiplication of reals in a given interval $[-k,k]$ down in the interval $[0,1]$; the details of this are given in Appendix~\ref{appendix: auxiliary connectives}. We then prove the theorem by giving recursive translations from terms of \coolname to formulae of $\RPLM$ and $\LPHF$, and vice versa. The correctness of the translations is proven by induction. The details are given in Appendix~\ref{appendix: Characterising coolname}.
\end{proof}

By restricting the degree of the terms and the fragment of the logic in the above theorem and carrying the restriction through the proof, we obtain the following fuzzy logic characterisation of \coolname with degree bound $d \in \N$ via the logics $\RPLM_{\leq d}$ and $\LPHF_{\leq d}$.

\optionaltheorem

So far in relation to both \coolname and $\RPL$, we have restricted to the case where constant parameters and truth constants are rational numbers. Analogously to \coolname and $\RPL$, we can define \coolnameb and 
Pavelka Logic ($\RPLb$) 
by replacing the rationals in the definitions of \coolname and $\RPL$ with reals. All the same notions and lemmas from Sections~\ref{section: Preliminaries} and \ref{section: Equivalence} and Appendix~\ref{appendix: auxiliary connectives} relating to \coolname and $\RPL$ can then analogously be defined and obtained for \coolnameb and $\RPLb$, including real-parameter neural networks, the logic $\RPLMb$ and the fragments $\RPLMb_{\leq n}$ for each $n \in \N$. 
We also note that the proofs of Theorems~\ref{theorem: R-algebra = RPLM = LPHF} and \ref{theorem: coolname with degree bound = RPLM and LPHF with deree bound} did not hinge upon the restriction to rationals.
Thus, we obtain the following corollaries.

\begin{corollary}
    \coolnameb has the same expressive power (w.r.t. scaling) as $\RPLMb$.
\end{corollary}

\begin{corollary}
    For each $d \in \N$, \coolnameb with degree bound $d$ has the same expressive power (w.r.t. scaling) as $\RPLMb_{\leq d}$.
\end{corollary}

We primarily considered \coolname rather than \coolnameb because in the case of rationals we also obtain a characterisation via $\LPHF$. With reals, this would require making the vocabulary of $\LPHF$ nondenumerable by, e.g., adding a truth constant for each real number in $[0,1]$, because even with a denumerable vocabulary the number of distinct formulae would be only denumerable, i.e., not even enough to express every real in $[0,1]$. 
However, the core appeal of $\LPHF$ over $\RPLM$ is that the vocabulary of $\LPHF$ is essentially finite whereas the vocabulary of $\RPLM$ is denumerable. Accordingly, the appeal of $\RPLM$ over $\LPHF$ is that due to the richer vocabulary, a formula of $\RPLM$ is sometimes significantly shorter than the shortest equivalent formula of $\LPHF$.

\subsection{Characterising neural networks}\label{section: Characterising neural networks}

In this section, we give fuzzy logic characterisations of neural networks.

\secondtheorem
\begin{proof}
    (Sketch) We first give a recursive constructive translation from formulae of $\RPLM_{\leq 1}$ and $\LPHF_{\leq 1}$ into neurons. Then, we give a recursive constructive translation from proto-neurons to formulae of $\RPL$ and $\LPHFF$; this direction involves inflating the proto-neurons such that they only contain integers, after which multiplication by a constant can be simulated in the logic via addition. The correctness of the translation is shown via induction. The details are in Appendix~\ref{appendix: Characterising neural networks}.
\end{proof}

For examples of how neurons translate to formulae and vice versa, see Appendix~\ref{appendix: Characterising neural networks}.

In the proof of Theorem~\ref{theorem: neural networks = RPLM1 = LPHF1}, the translation from $\RPLM_{\leq 1}$ to neurons does not hinge on the constant symbols in the neurons and the truth constants in the formulae being rationals; the same translation works with reals. The translation from proto-neurons to $\RPL$ does depend on this restriction, but if we only want to translate neurons into $\RPLM_{\leq 1}$ instead of $\RPL$, then we can use Theorem~\ref{theorem: coolname with degree bound = RPLM and LPHF with deree bound} where the proof does not depend on restricting to rationals. Thus, we obtain the following corollary.

\begin{corollary}
    Real-parameter neural networks have the same expressive power (w.r.t. scaling) as $\RPLMb_{\leq 1}$.
\end{corollary}

Again, the reason we focused on neural networks with rational weights and biases was because we were also able to obtain characterisations via $\RPL$, $\LPHF_{\leq 1}$ and $\LPHFF$. For reals, we could obtain a characterisation via adding a truth constant for each real in $[0,1]$ to $\LPHF_{\leq 1}$, whereas for $\RPL$ and $\LPHFF$ even this would not suffice to obtain a characterisation.
The characterisation via $\RPL$ is especially interesting, since $\RPL$ is known to enjoy Pavelka-style completeness \cite{hajek1995fuzzy}.

From Theorem~\ref{theorem: neural networks = RPLM1 = LPHF1}, we obtain a corollary between $\RPLM_{\leq 1}$, $\LPHF_{\leq 1}$, $\RPL$ and $\LPHFF$.

\begin{corollary}
    $\RPLM_{\leq 1}$ and $\LPHF_{\leq 1}$ have the same expressive power (w.r.t. scaling) as $\RPL$ and $\LPHFF$.
\end{corollary}
\begin{proof}
    Each formula of $\RPL$ is a formula of $\RPLM_{\leq 1}$ and by Proposition~\ref{proposition: RPLM_d = LPHF_d} there exists an equivalent formula of $\LPHF_{\leq 1}$. Similarly, each formula of $\LPHFF$ is a formula of $\LPHF_{\leq 1}$ and equivalent to some formula of $\RPLM_{\leq 1}$. For the converse, Theorem~\ref{theorem: neural networks = RPLM1 = LPHF1} implies that each formula $\varphi$ of $\RPLM_{\leq 1}$ or $\LPHF_{\leq 1}$ is equivalent to a neuron $n$ which in turn is $(1,k)$-equivalent to a formula $\psi$ of $\RPL$ or $\LPHFF$ for some $k > 1$; by Lemma~\ref{lemma: transitive equivalence}, $\psi$ is $k$-equivalent to $\varphi$.
\end{proof}

In Appendix~\ref{appendix: lukasiewicz logic}, we consider some proto-neurons that cannot be translated into $(i,k)$-equivalent formulae of {\L}ukasiewicz logic. We also study what kinds of neurons are expressible in {\L}ukasiewicz logic.

\section{Conclusion}

We have given fuzzy logic characterisations of $\ReLU$-activated rational-weight neural networks via fragments of two fuzzy logics, $\RPLM$ and $\LPH$. We also gave a fuzzy logic characterisation of a more general structure consisting of polynomials from the polynomial ring over $\Q$ in countably many variables appended with $\ReLU$. The inputs for the neural networks and for the polynomials are allowed to be arbitrary real numbers. 
Future work could involve studying how neural networks are connected to tautologies of the characterising fuzzy logics, as well as studying the completeness properties of the fuzzy logic fragments introduced here.

Logic‑based approaches to AI systems are useful for, e.g., verification and explainability. Verification is important in applications with safety‑related components. Explainability, on the other hand, concerns the human interface: it is precisely in such contexts that AI replacing humans may be unlikely, since human understanding itself serves as the primary point of reference. Moreover, explainability is also central in scientific uses of AI, where it is typically not sufficient to obtain results without understanding the reasons behind them. Explainability is thus clearly important and naturally motivates the study of logics that are cognitively easy to process. Whether fuzzy logic satisfies this requirement is debatable, but investigating different kinds of logics helps to clarify the overall landscape. Furthermore, logical methods can also provide new perspectives on training AI systems, potentially leading to useful algorithmic approaches that differ from standard gradient‑descent methods.

\section{Acknowledgments}
    
Damian Heiman was supported by the Magnus Ehrnrooth Foundation. Antti Kuusisto was supported by the project \emph{Perspectives on computational logic}, funded by the Research Council of Finland, project number 369424.

\bibliography{references}

\appendix

\section{Proof of Lemma~\ref{lemma: proto-neuron degrees}}\label{appendix: neural networks}

In this section, we prove Lemma~\ref{lemma: proto-neuron degrees}:

\degreelemma*
\begin{proof}
    We prove the claim by induction over the structure of $t$. If $t = \overline{r}$ for some $r \in \Q$ or if $t = x$ for some $x \in \VAR$, then $t$ is a proto-neuron by definition. Next, assume that $t'$ and $t''$ are terms of \coolname and that $t'$ and $t''$ are proto-neurons if and only if $\deg(t') \leq 1$ and $\deg(t'') \leq 1$, respectively. 
    \begin{itemize}
        \item If $t = t' + t''$, then $t$ is a proto-neuron if and only if $t'$ and $t''$ are proto-neurons. By the induction hypothesis, this is equivalent to $\deg(t'), \deg(t'') \leq 1$ which in turn is equivalent to $\deg(t) = \max\{\deg(t'), \deg(t'')\} \leq 1$.
        \item If $t = t't''$, then $t$ is a proto-neuron if and only if $t'$ and $t''$ are proto-neurons such that we have $\deg(t') + \deg(t'') \leq 1$. Since $\deg(t) = \deg(t') + \deg(t'')$, this is equivalent to $\deg(t) \leq 1$.
        \item If $t = \ReLU(t')$, then $t$ is a proto-neuron if and only if $t'$ is a proto-neuron. By the induction hypothesis, this is equivalent to $\deg(t') \leq 1$, and since $\deg(t) = \deg(t')$, this is equivalent to $\deg(t) \leq 1$.
    \end{itemize}
    This concludes the induction.
\end{proof}

\section{Defining rational numbers in $\LPH$}\label{appendix: defining rationals}

It is shown in \cite{esteva2001} that all rational numbers in $[0,1]$ are definable in $\LPH$, meaning that for every $r \in \Q \cap [0,1]$ there exists a formula $\varphi$ of $\LPH$ such that $\lmodel(\varphi) = r$ for each valuation $\lmodel$. We show an alternative construction to the one in \cite{esteva2001}, but the results of the current paper hold regardless of which way the construction is done. We define each rational number $r \in \Q \cap [0,1]$ in $\LPH$ as a formula $\underline{r}$ recursively as follows.
First, we can obviously define $\underline{0} = \overline{0}$ and $\underline{\half} = \overline{\half}$. Second, we define $\underline{1} := \negl \overline{0}$.
Next, we cover the case where $r = \frac{k}{2^j} \in ]0,1[$ for some $k, j \in \Z_+$ such that $k < 2^j$.
If $k = j = 1$, then $r = \half$, and thus $\underline{r}$ is already defined. Next, assume that $j > 1$ and we have defined $\underline{\frac{k'}{2^{j-1}}}$ for each $k' \in \Z_+$ such that $k' < 2^{j-1}$. If $k = 1$, then we define
\[
    \textstyle{\underline{\frac{1}{2^{j}}} := \underline{\frac{1}{2^{j-1}}} \odot \underline{\half}}.
\]
Next, assume that $k > 1$ and we have defined $\underline{\frac{k'}{2^j}}$ for each $k' \in \Z_+$ such that $k' < k$. If $k$ is even, then $\frac{k}{2^j} = \frac{k'}{2^{j-1}}$ where $k' = \frac{k}{2} \in \Z_+$, meaning that $\underline{r}$ is already defined. If $k$ is odd, then $2^j$ is not divisible by $k$, meaning that $\underline{r}$ is not yet defined. We define
\[
    \textstyle{\underline{\frac{k}{2^{j}}} := \underline{\frac{k-1}{2^{j}}} \oplus \underline{\frac{1}{2^{j}}}.}
\]
Lastly, we consider the case where $r = \frac{k}{\ell} \in ]0,1[$ for some $k, \ell \in \Z_+$ such that $k < \ell$ and $\ell$ is not divisible by $k$ nor is $\ell$ a power of $2$. Let $m$ be the smallest integer such that $k, \ell < 2^m$. Then
\[
    \textstyle{\underline{\frac{k}{\ell}} := \underline{\frac{\ell}{2^m}} \impp \underline{\frac{k}{2^m}}.}
\]
This completes the construction.

\begin{lemma}\label{lemma: truth constants in LPH}
    For each $r \in \Q \cap [0,1]$ and each valuation $\lmodel$, we have $\lmodel(\underline{r}) = r$.
\end{lemma}
\begin{proof}
    By induction over the recursive definition of $\underline{r}$.
\end{proof}

\section{Auxiliary connectives}\label{appendix: auxiliary connectives}

In this section, we establish auxiliary connectives that will be useful when translating terms of \coolname to the logics $\RPLM$ and $\LPH$.
The goal with these auxiliary connectives is to simulate real arithmetic from some arbitrary interval $[-k,k]$ in the interval $[0,1]$. While we already have the connectives $\oplus$ and $\odot$ whose semantics are similar to the sum and multiplication of real numbers, this does not yet suffice for our characterisations. Consider for example the multiplication of two real numbers $a,b > 1$. We have $ab > a$ and $ab > b$. On the other hand, multiplying any two numbers $a,b \in [0,1]$, we get $ab \leq a$ and $ab \leq b$. Thus, $\odot$ by itself does not suffice to simulate multiplication of two real numbers, but we will see by the end of the section that a more involved auxiliary connective will accomplish the task.

To begin with, we first define for each $n \in \N$ the iterated sum operator $\itsum_n$. Intuitively, $\itsum_n$ corresponds to multiplying the truth value of a formula by $n$, which is achieved by iteratively summing the formula with itself. For $n = 0$ we define $\itsum_0 \varphi := \overline{0}$. Then, for each $n \in \N$, we define
\[
    \textstyle{\itsum_{n+1} \varphi := (\itsum_{n} \varphi) \oplus \varphi.}
\]

\begin{lemma}\label{lemma: scalar multiplication formula}
    For all $n \in \N$, all valuations $\lmodel$ and all formulae $\varphi$ of $\RPLM$ or $\LPH$, we have that $\lmodel(\itsum_n \varphi) = \min\{1, n\lmodel(\varphi)\}$. 
\end{lemma}
\begin{proof}
    By induction over $n$.
\end{proof}

For $\RPLM$ and its fragments, we introduce the connective $\overline{\oplus}'$.
The intuition behind $\overline{\oplus}'$ is to simulate the addition of real numbers belonging to the interval $[-k,k]$ for any $k \in \R_+$ in the interval $[0,1]$ using the scaling function $\scale_k$.
More formally, we define $\overline{\oplus}'$ as follows:
\[
    \textstyle{\varphi \,\,\overline{\oplus}' \, \psi := \left(\,\varphi \ominus \overline{\frac{1}{4}}\,\right) \oplus \left(\,\psi \ominus \overline{\frac{1}{4}}\,\right).}
\]
For $\LPH$ and its fragments, where the truth constant $\overline{\frac{1}{4}}$ does not exist, we instead define the connective $\underline{\oplus}'$ which is defined the same as $\overline{\oplus}'$ except that we replace $\overline{\frac{1}{4}}$ with the formula~$\underline{\frac{1}{4}}$ (recall Appendix~\ref{appendix: defining rationals}).
In the rest of the paper, the logic is always clear from the context, and thus we just write $\oplus'$ instead of $\overline{\oplus}'$ or $\underline{\oplus}'$.

The following lemma shows that $\oplus'$ indeed simulates the addition of two real numbers in the interval $[-k,k]$ correctly, assuming that the inputs lie in the interval $\left[-\frac{k}{2},\frac{k}{2}\right]$. The restriction of inputs is natural, since $a,b \in \left[-\frac{k}{2},\frac{k}{2}\right]$ implies $a+b \in [-k,k]$.

\begin{lemma}\label{lemma: addition formula}
    Let $\cL$ be $\RPLM$ or $\LPH$.
    For all valuations $\lmodel$, all $k \in \R_+$ and all formulae $\varphi$ and $\psi$ of $\cL$: if $\lmodel(\varphi) = \scale_k(\ell)$ and $\lmodel(\psi) = \scale_k(\ell')$ for some $\ell, \ell' \in \left[-\frac{k}{2}, \frac{k}{2}\right]$, then $\lmodel(\varphi \oplus' \psi) = \scale_k(\ell + \ell')$. 
\end{lemma}
\begin{proof}
    We cover the case where $\cL$ is $\RPLM$ as the case where $\cL$ is $\LPH$ is analogous.
    By the semantics of $\ominus$, we have $\lmodel\left(\varphi \ominus \overline{\frac{1}{4}}\right) = \max\left\{0, \lmodel(\varphi) - \lmodel\left(\overline{\frac{1}{4}}\right)\right\}$. Because $\lmodel\left(\overline{\frac{1}{4}}\right) = \frac{1}{4}$ and  
    \[
        \textstyle{\lmodel(\varphi) = \scale_k(\ell) \geq \scale_k(-\frac{k}{2}) = \frac{k - \frac{k}{2}}{2k} = \frac{1}{4},}
    \]
    we have $\lmodel(\varphi) - \lmodel\left(\overline{\frac{1}{4}}\right) \geq \frac{1}{4} - \frac{1}{4} = 0$. Thus, we have $\lmodel\left(\varphi \ominus \overline{\frac{1}{4}}\right) = \lmodel(\varphi) - \frac{1}{4}$. Analogously, we see that $\lmodel\left(\psi \ominus \overline{\frac{1}{4}}\right) = \lmodel(\psi) - \frac{1}{4}$.

    Now, by the semantics of $\oplus$ we have $\lmodel(\varphi \oplus' \psi) = \min\left\{1, \lmodel\left(\varphi \ominus \overline{\frac{1}{4}}\right) + \lmodel\left(\psi \ominus \overline{\frac{1}{4}}\right)\right\}$. By the above, we have $\lmodel\left(\varphi \ominus \overline{\frac{1}{4}}\right) + \lmodel\left(\psi \ominus \overline{\frac{1}{4}}\right) = \left(\lmodel(\varphi) - \frac{1}{4}\right) + \left(\lmodel(\psi) - \frac{1}{4}\right) = \lmodel(\varphi) + \lmodel(\psi) - \half$, which implies $\lmodel(\varphi \oplus' \psi) = \min\left\{1, \lmodel(\varphi) + \lmodel(\psi) - \half\right\}$. Because $\ell, \ell' \leq \frac{k}{2}$, we have 
    \[
        \textstyle{\lmodel(\varphi), \lmodel(\psi) \leq \scale_k\left(\frac{k}{2}\right) = \frac{k+\frac{k}{2}}{2k} = \frac{3}{4}},
    \]
    which means that
    \[
        \textstyle{\lmodel(\varphi) + \lmodel(\psi) - \half \leq \frac{3}{4} + \frac{3}{4} - \half = 1.}
    \]
    Thus 
    \[
        \begin{aligned}
            \textstyle{\lmodel(\varphi \oplus' \psi)\,} &\textstyle{= \lmodel(\varphi) + \lmodel(\psi) - \half} \\
            &\textstyle{= \scale_k(\ell) + \scale_k(\ell') - \half} \\
            &\textstyle{= \frac{k + \ell}{2k} + \frac{k + \ell'}{2k} - \frac{k}{2k}} \\
            &\textstyle{= \frac{k + (\ell + \ell')}{2k}} \\
            &\textstyle{= \scale_k(\ell + \ell').}
        \end{aligned}
    \]
    This completes the proof.
\end{proof}

\begin{remark}
    The definition of $\oplus'$ may raise the following question: why not simply define $\varphi \oplus' \psi := (\varphi \oplus \psi) \ominus \overline{\half}$? The problem is that, while the formula would still give correct truth values when $\lmodel(\varphi) = \scale_k(\ell)$ and $\lmodel(\psi) = \scale_k(\ell')$
    for some $\ell, \ell' \in \left[-\frac{k}{2}, 0\right]$, 
    it would not give correct values when $\ell, \ell' \in \left]0, \frac{k}{2}\right]$, because in all such cases we would get 
    $\lmodel(\varphi \oplus' \psi) = \half = \scale_k(0) \neq \scale_k(\ell + \ell')$.
    The chosen definition of $\oplus'$ is intended to maximize the working range of the connective while also treating the two connected formulae $\varphi$ and $\psi$ symmetrically.
\end{remark}

Analogously to the connective $\itsum_n$, we define a second iterated sum operator $\itsum'_n$ for each $n \in \N$, this time corresponding to the connective $\oplus'$. Again, the operator $\itsum'_n$ corresponds to multiplying by an integer $n$, but this time the multiplication occurs in the larger interval $[-k,k]$ and the result is scaled down to $[0,1]$. For $n = 0$, we define $\itsum'_0 \varphi := \overline{\half}$, and for $n = 1$, we define $\itsum'_1 \varphi := \varphi$. 
Next, let $n > 1$ and assume we have defined $\itsum'_{m} \varphi$ for each non-negative integer $m < n$. 
We define
\[
    \textstyle{
        \itsum'_{n} \varphi := (\itsum'_{\lceil\frac{n}{2}\rceil} \varphi) \oplus' (\itsum'_{\lfloor\frac{n}{2}\rfloor} \varphi).
    }
\]

The next lemma shows that $\itsum'_{n}$ correctly simulates the multiplication of a real number in the interval $[-k,k]$ by $n$, assuming that the input is in the interval $\left[-\frac{k}{n+1}, \frac{k}{n+1}\right]$.
This restriction is natural, since 
$a \in \left[-\frac{k}{n+1}, \frac{k}{n+1}\right]$ implies $na \in \left[-\frac{nk}{n+1}, \frac{nk}{n+1}\right] \subset [-k,k]$.

\begin{lemma}\label{lemma: semantics of iterated plus'}
    Let $\cL$ be $\RPLM$ or $\LPH$.
    Let $n \in \N$, $k \in \R_+$, let $\lmodel$ be a valuation and $\varphi$ a formula of $\cL$ such that $\lmodel(\varphi) = \scale_k(\ell)$ for some $\ell \in \left[-\frac{k}{n+1}, \frac{k}{n+1}\right]$. Then we have
    $\lmodel(\itsum'_n \varphi) = \scale_k(n \ell)$. 
\end{lemma}
\begin{proof}
    We prove the claim by induction over $n$. If $n = 0$, then $\itsum'_0 \varphi = \overline{\half}$ and thus we have $\lmodel(\itsum'_0 \varphi) = \half = \scale_k(0 \ell)$. If $n = 1$, then $\itsum'_1 \varphi = \varphi$ and $\lmodel(\itsum'_1 \varphi) = \lmodel(\varphi) = \scale_k(1\ell)$. Next, assume $n > 1$ and the claim holds for each $n' < n$. Assume that $\lmodel(\varphi) = \scale_k(\ell)$ for some $\ell \in \left[-\frac{k}{n+1}, \frac{k}{n+1}\right]$. By the definition of $\itsum_n'$, we have $\itsum_n' \varphi = (\itsum'_{\lceil\frac{n}{2}\rceil} \varphi) \oplus' (\itsum'_{\lfloor\frac{n}{2}\rfloor} \varphi)$.
    Because $\ell \in \left[-\frac{k}{n+1}, \frac{k}{n+1}\right]$ implies $\ell \in \left[-\frac{k}{\lceil\frac{n}{2}\rceil+1}, \frac{k}{\lceil\frac{n}{2}\rceil+1}\right]$ and $\ell \in \left[-\frac{k}{\lfloor\frac{n}{2}\rfloor+1}, \frac{k}{\lfloor\frac{n}{2}\rfloor+1}\right]$, by the induction hypothesis we have $\lmodel(\itsum'_{\lceil\frac{n}{2}\rceil} \varphi) = \scale_k(\lceil\frac{n}{2}\rceil\ell)$ and $\lmodel(\itsum_{\lfloor\frac{n}{2}\rfloor}' \varphi) = \scale_k(\lfloor\frac{n}{2}\rfloor\ell)$. We want to use Lemma~\ref{lemma: addition formula}, for which we need to show that $\lceil\frac{n}{2}\rceil\ell, \lfloor\frac{n}{2}\rfloor\ell \in \left[-\frac{k}{2},\frac{k}{2}\right]$.

    Because of the fact that $\lceil\frac{n}{2}\rceil, \lfloor\frac{n}{2}\rfloor \leq \frac{n+1}{2}$ and the assumption that $\ell \in \left[-\frac{k}{n+1}, \frac{k}{n+1}\right]$, we have $\lceil\frac{n}{2}\rceil\ell, \lfloor\frac{n}{2}\rfloor\ell \in \left[\frac{n+1}{2}(-\frac{k}{n+1}), \frac{n+1}{2}\frac{k}{n+1}\right] = \left[-\frac{k}{2}, \frac{k}{2}\right]$.
    Now by Lemma~\ref{lemma: addition formula}, we must have $\lmodel(\itsum_n \varphi) = \scale_k(\lceil\frac{n}{2}\rceil\ell+\lfloor\frac{n}{2}\rfloor\ell) = \scale_k(n\ell)$, as desired.
\end{proof}

Lastly, 
for $\RPLM$ and its fragments, 
we define the connective $\overline{\odot}^k$, which is tied to an integer parameter $k \in \Z_+$. 
The intuition behind $\overline{\odot}^k$ is to simulate the multiplication of two real numbers belonging to the interval $[-k,k]$ in the interval $[0,1]$. 
Formally, 
$\overline{\odot}^k$ is defined as follows:
\[
    \textstyle{\varphi \,\,\overline{\odot}^k \, \psi := \itsum_k\left(\left(\itsum_2(\varphi \odot \psi) \oplus \overline{\frac{1}{2k}}\right) \ominus (\varphi \oplus' \psi)\right).}
\]
Again, for $\LPH$ and its fragments, we instead define $\underline{\odot}^k$ whose definition is otherwise the same as $\overline{\odot}^k$ except that we again replace the truth constant $\overline{\frac{1}{2k}}$ with the formula $\underline{\frac{1}{2k}}$. As with $\oplus'$, for the rest of the paper the logic is always clear from the context, and so we simply write $\odot^k$. 

The next lemma shows that $\odot^k$ indeed correctly simulates the multiplication of reals in $[-k,k]$, assuming the inputs are in the interval $\bigl[-\sqrt{k}, \sqrt{k}\,\bigr]$. The restriction of inputs is once again natural, since $a,b \in \bigl[-\sqrt{k}, \sqrt{k}\,\bigr]$ implies $ab \in [-k,k]$.

\begin{lemma}\label{lemma: multiplication formula}
    Let $\cL$ be $\RPLM$ or $\LPH$.
    For all integers $k \geq 8$, valuations~$\lmodel$ and formulae $\varphi$ and $\psi$ of $\cL$: if $\lmodel(\varphi) = \scale_k(\ell)$ and $\lmodel(\psi) = \scale_k(\ell')$ for some $\ell, \ell' \in \bigl[-\sqrt{k}, \sqrt{k}\,\bigr]$, then $\lmodel(\varphi \odot^k \psi) = \scale_k(\ell \ell' )$. 
\end{lemma}
\begin{proof}
    We show the case where $\cL$ is $\RPLM$ as the case where $\cL$ is $\LPH$ is analogous.
    For convenience, we let $\theta_{\text{left}} := \itsum_2(\varphi \odot \psi) \oplus \overline{\frac{1}{2k}}$ and $\theta_{\text{right}} := \varphi \oplus' \psi$. This means that $\varphi \odot^k \psi = \itsum_k(\theta_{\text{left}} \ominus \theta_{\text{right}})$. 
    We start by calculating $\lmodel(\theta_{\text{left}})$, for which we must first calculate $\lmodel(\itsum_2(\varphi \odot \psi))$ and $\lmodel(\varphi \odot \psi)$.
    First, by definition we get
    \[
        \textstyle{\lmodel(\varphi \odot \psi) = \lmodel(\varphi) \lmodel(\psi) = \scale_k(\ell) \scale_k(\ell') = \frac{k+\ell}{2k}\frac{k+\ell'}{2k} = \frac{k^2 + k(\ell + \ell') + \ell\ell'}{4k^2}. }
    \]
    By Lemma~\ref{lemma: scalar multiplication formula}, we have $\lmodel(\itsum_2(\varphi \odot \psi)) = \min\{1, 2\lmodel(\varphi \odot \psi)\}$. Because $\ell, \ell' \in \bigl[-\sqrt{k}, \sqrt{k}\,\bigr]$ we find that 
    \[
        \begin{aligned}
            \textstyle{2 \lmodel(\varphi \odot \psi)\,} &\textstyle{= 2 \frac{k^2 + k(\ell + \ell') + \ell\ell'}{4k^2}} \\
            &\textstyle{= \frac{k^2 + k(\ell + \ell') + \ell\ell'}{2k^2}} \\
            &\textstyle{\leq \frac{k^2 + k\big{(}\sqrt{k} + \sqrt{k}\big{)} + \sqrt{k}\sqrt{k}}{2k^2}} \\
            &\textstyle{= \frac{k^2 + 2k\sqrt{k} + k}{2k^2}} \\
            &\textstyle{= \frac{k+2\sqrt{k}+1}{2k},}
        \end{aligned}
    \]
    which is at most $1$ when $k \geq 3 + 2\sqrt{2} \approx 5.83$. 
    Because of the assumption that $k \geq 8$, we have $2 \lmodel(\varphi \odot \psi) \leq 1$ and thus $\lmodel(\itsum_2(\varphi \odot \psi)) = \frac{k^2 + k(\ell + \ell') + \ell\ell'}{2k^2}$. Next, by the semantics of~$\oplus$ we have $\lmodel\left(\theta_{\text{left}}\right) = \min\left\{1, \lmodel(\itsum_2(\varphi \odot \psi)) + \lmodel\left(\overline{\frac{1}{2k}}\right)\right\}$. Utilising the analysis above, we find that 
    \[
        \begin{aligned}
            \textstyle{\lmodel(\itsum_2(\varphi \odot \psi)) + \lmodel\left(\,\overline{\frac{1}{2k}}\,\right)\,} &\textstyle{= \frac{k^2 + k(\ell + \ell') + \ell\ell'}{2k^2} + \frac{1}{2k}} \\
            &\textstyle{\leq \frac{k+2\sqrt{k}+1}{2k} + \frac{1}{2k}} \\
            &\textstyle{= \frac{k+2\sqrt{k}+2}{2k}},
        \end{aligned}
    \]
    which is at most $1$ when $k \geq 4 + 2\sqrt{3} \approx 7.46$. Again, because of the assumption $k \geq 8$, we see that $\lmodel(\itsum_2(\varphi \odot \psi)) + \lmodel\left(\,\overline{\frac{1}{2k}}\,\right) \leq 1$ and thus 
    \[
        \textstyle{\lmodel\left(\theta_{\text{left}}\right) = \frac{k^2 + k(\ell + \ell') + \ell\ell'}{2k^2} + \frac{1}{2k} = \frac{k^2 + k(\ell + \ell') + k + \ell\ell'}{2k^2}.}
    \]
    Next, we calculate $\lmodel(\theta_{\text{right}})$.
    We see that $\sqrt{k} \leq \frac{k}{2}$ when $k \geq 4$ and thus $\ell, \ell' \in \bigl[-\sqrt{k}, \sqrt{k}\,\bigr]$ implies $\ell, \ell' \in \left[-\frac{k}{2}, \frac{k}{2}\right]$, which means we can apply Lemma~\ref{lemma: addition formula} to $\theta_{\text{right}}$, giving us 
    \[
        \textstyle{\lmodel(\theta_{\text{right}}) = \scale_k(\ell + \ell') = \frac{k + (\ell + \ell')}{2k} = \frac{k^2 + k(\ell + \ell')}{2k^2}.}
    \]
    We finish the proof by calculating $\lmodel(\varphi \odot^k \psi)$, for which we first calculate $\lmodel(\theta_{\text{left}} \ominus \theta_{\text{right}})$.
    By the semantics of $\ominus$, we have $\lmodel(\theta_{\text{left}} \ominus \theta_{\text{right}}) = \max\{0, \lmodel(\theta_{\text{left}}) - \lmodel(\theta_{\text{right}})\}$. We see that
    \[
        \textstyle{\lmodel(\theta_{\text{left}}) - \lmodel(\theta_{\text{right}}) = \frac{k^2 + k(\ell + \ell') + k + \ell\ell'}{2k^2} - \frac{k^2 + k(\ell + \ell')}{2k^2} = \frac{k + \ell\ell'}{2k^2} \geq \frac{k -\sqrt{k}\sqrt{k}}{2k^2} = 0,}
    \]
    which means that $\lmodel(\theta_{\text{left}} \ominus \theta_{\text{right}}) = \lmodel(\theta_{\text{left}}) - \lmodel(\theta_{\text{right}}) = \frac{k + \ell\ell'}{2k^2}$.
    Lastly, by Lemma~\ref{lemma: scalar multiplication formula}, we have $\lmodel(\varphi \odot^k \psi) = \min\{1, k\,\lmodel(\theta_{\text{left}} \ominus \theta_{\text{right}})\} = \min\left\{1, \frac{k + \ell\ell'}{2k}\right\}$. Because $\ell, \ell' \in \bigl[-\sqrt{k}, \sqrt{k}\,\bigr]$, we get $\frac{k + \ell\ell'}{2k} \leq \frac{k + \sqrt{k}\sqrt{k}}{2k} = 1$, which gives us $\lmodel(\varphi \odot^k \psi) = \frac{k + \ell\ell'}{2k} = \scale_k(\ell \ell')$, as desired.
\end{proof}

The following lemma extends Lemma~\ref{lemma:degrees of formulae} to the auxiliary connectives defined in this section, relating them to the fragments of $\RPLM$ and $\LPHF$ defined in Section~\ref{section: fragments}. Importantly, it shows that the scaled sum $\oplus'$ and multiplication $\odot^{k}$ do not differ from the ordinary sum $\oplus$ and multiplication $\odot$ in terms of how they relate to the fragments.

\begin{lemma}\label{lemma: degrees of special formulae}
    Let $\cL$ be $\RPLM$ or $\LPHF$.
    Let $i,j \in \N$, and let $\varphi$ be a formula of $\cL_{\leq i}$ and $\psi$ a formula of $\cL_{\leq j}$. Then the following hold.
    \begin{itemize}
        \item For each $n \in \N$, $\itsum_n \varphi$ and $\itsum'_n \varphi$ are formulae of $\cL_{\leq i}$.
        \item $\varphi \oplus' \psi$ is a formula of $\cL_{\leq \max\{i, j\}}$.
        \item For each $k \in \Z_+$, $\varphi \odot^k \psi$ is a formula of $\cL_{\leq i + j}$.
    \end{itemize}
\end{lemma}
\begin{proof}
    The proof is straightforward by the definitions of the operators and the fragments $\cL_{\leq n}$ and by Lemma~\ref{lemma:degrees of formulae}.
\end{proof}

\section{Proof of Theorem~\ref{theorem: R-algebra = RPLM = LPHF}}\label{appendix: Characterising coolname}

In this section, we prove Theorem~\ref{theorem: R-algebra = RPLM = LPHF}. First, we recall the theorem. 

\firsttheorem*

We start by defining two auxiliary notions regarding terms of \coolname: atomic length and subterm. Intuitively, subterms are those components of a term that are themselves terms, and atomic length refers to how many of these subterms are atomic (i.e., either rational constants or variables).
More formally, we define the \defstyle{atomic length} of a term $t$, denoted by $\length(t)$, as follows: 
\begin{itemize}
    \item If $t = \overline{r}$ for some $r \in \Q$ or $t = x$ for some $x \in \VAR$, then $\length(t) = 1$.
    \item If $t = t' + t''$ or $t = t' t''$ for terms $t'$ and $t''$, then $\length(t) = \length(t') + \length(t'')$.
    \item If $t = \ReLU(t')$ for some term $t'$, then $\length(t) = \length(t')$.
\end{itemize}
Likewise, we define the set $\subt(t)$ of \defstyle{subterms} of a term $t$ as follows: 
\begin{itemize}
    \item If $t = \overline{r}$ for some $r \in \Q$ or $t = x$ for some $x \in \VAR$, then $\subt(t) = \{t\}$.
    \item If $t = t' + t''$ or $t = t' t''$ for terms $t'$ and $t''$, then $\subt(t) = \subt(t') \cup \subt(t'') \cup \{t\}$.
    \item If $t = \ReLU(t')$ for some term $t'$, then $\subt(t) = \subt(t') \cup \{t\}$.
\end{itemize}

Before proving the theorem, we establish two lemmas. In the first lemma, we give a constructive translation from terms of \coolname to formulae of $\RPLM$ and $\LPHF$.

\begin{lemma}\label{lemma: terms to logic}
    Let $\cL$ be $\RPLM$ or $\LPHF$.
    For each term $t$ of \coolname and each $i \in \R_+$, there exists some $K \in \Z_+$ such that for each integer $k \geq K$ we can construct an $(i,k)$-equivalent formula $\varphi_t$ of $\cL$. 
    For all $d \in \N$, if
    $\deg(t) \leq d$, 
    then $\varphi_t$ is a formula of~$\cL_{\leq d}$.
\end{lemma}
\begin{proof}
    The strategy is to recursively build an $(i,k)$-equivalent formula for each subterm of~$t$ including~$t$ itself.
    We show the proof for the case where $\cL$ is $\RPLM$. The case where $\cL$ is $\LPHF$ follows by Propositions~\ref{proposition: RPLM_d = LPHF_d} and \ref{proposition: RPLM = LPHF}.
    
    Let $t$ be a term and let $i \in \R_+$. Let $r_{\max}$ be the largest rational number such that $\overline{r_{\max}}$ or $\overline{-r_{\max}}$ is a subterm of $t$ and let $j := \max\{i, r_{\max}, 2\}$. For subterms $t'$ of $t$, it is easy to show by induction over the structure of $t'$ that given any evaluation map $\tmodel$ such that $\tmodel(x) \in [-i,i]$ for all $x \in \VAR$, we have $|\tmodel(t')| \leq j^{\length(t')}$. 
    Now let $K := \max\{j^{2\length(t)}, 8\}$ and let $k \geq K$ be an integer.

    For each subterm $s$ of $t$, we recursively construct a formula $\varphi_{s}$ that is $(i,k)$-equivalent to $s$ as follows:
    \begin{itemize}
        \item If $s = \overline{r}$ for some $r \in \Q$, then $|r| \leq r_{\max} \leq j \leq K \leq k$, which implies that $\scale_k(r) \in [0,1]$. Moreover, because $k, r \in \Q$, we must have $\scale_k(r) = \frac{k+r}{2k} \in \Q$ and therefore $\scale_k(r) \in \Q \cap [0,1]$. We define $\varphi_s := \overline{{\scale_k(r)}}$.
        \item If $s = x$ for some $x \in \VAR$, then $\varphi_s := p_x$. 
        \item If $s = t' + t''$ for some terms $t'$ and $t''$, then $\varphi_s := \varphi_{t'} \oplus' \varphi_{t''}$. 
        \item If $s = t' t''$ for some terms $t'$ and $t''$, then $\varphi_s := \varphi_{t'} \odot^k \varphi_{t''}$ (which is defined because $k \in \Z_+$). 
        \item If $s = \ReLU(t')$ for some term $t'$, then $\varphi_s := \left(\varphi_{t'} \ominus \overline{\half}\right) \oplus \overline{\half}$. 
    \end{itemize}

    We now prove the correctness of our construction. Let $s$ be a subterm of $t$, let $\tmodel$ be an evaluation map such that $\tmodel(x) \in [-i,i]$ for each $x \in \VAR$, and let $\lmodel$ be the valuation such that $\lmodel(p) = \scale_k(\tmodel(x_p))$ for each $p \in \PROP$. We have to show $\lmodel(\varphi_s) = \scale_k(\tmodel(s))$. We prove the claim by induction over the structure of $s$. First, if $s = \overline{r}$ for some $r \in \Q$, then $\lmodel(\varphi_s) = \lmodel\left(\overline{\scale_k(r)}\right) = \scale_k(r) = \scale_k(\tmodel(s))$.
    Next, if $s = x$ for some $x \in \VAR$, then the claim is true by the definition of the valuation $\lmodel$.

    Next, assume that the claim holds for the terms $t'$ and $t''$, i.e., $\lmodel(\varphi_{t'}) = \scale_k(\tmodel(t'))$ and $\lmodel(\varphi_{t''}) = \scale_k(\tmodel(t''))$. Before delving into further steps, we show that the conditions of Lemmas~\ref{lemma: addition formula} and \ref{lemma: multiplication formula} are met. First, we already know that $|\tmodel(t')| \leq j^{\length(t')}$. Because $\length(t') \leq \length(t)$, we have $|\tmodel(t')| \leq j^{\length(t)}$. Since $j^{\length(t)} = \sqrt{j^{2\length(t)}}$, we get $|\tmodel(t')| \leq \sqrt{j^{2\length(t)}}$. Now, because $k \geq K \geq j^{2\length(t)}$, we have $|\tmodel(t')| \leq \sqrt{k}$. Analogously,  we see that $|\tmodel(t'')| \leq \sqrt{k}$. This means that $\tmodel(t'), \tmodel(t'') \in \bigl[-\sqrt{k}, \sqrt{k}\,\bigr]$. Furthermore, since $k \geq K \geq 8 > 4$, we have $\sqrt{k} < \frac{k}{2} < k$, which means that $\tmodel(t'), \tmodel(t'') \in \left[-\frac{k}{2}, \frac{k}{2}\right]$ and $\tmodel(t'), \tmodel(t'') \in \left[-k, k\right]$. Now, we are ready to cover the remaining cases.
    \begin{itemize}
        \item If $s = t' + t''$, then $\varphi_s = \varphi_{t'} \oplus' \varphi_{t''}$. 
        Because it is the case that $\tmodel(t'), \tmodel(t'') \in \left[-\frac{k}{2}, \frac{k}{2}\right]$, we may apply Lemma~\ref{lemma: addition formula}, which gives us
        $\lmodel(\varphi_s) = \scale_k(\tmodel(t')+\tmodel(t''))$. Now because $\tmodel(t')+\tmodel(t'') = \tmodel(t' + t'') = \tmodel(s)$, we have $\lmodel(\varphi_s) = \scale_k(\tmodel(s))$. 
        \item If $s = t' t''$, then $\varphi_s = \varphi_{t'} \odot^{k} \varphi_{t''}$. 
        Because it is the case that $\tmodel(t'), \tmodel(t'') \in \bigl[-\sqrt{k}, \sqrt{k}\,\bigr]$, by Lemma~\ref{lemma: multiplication formula}
        $\lmodel(\varphi_{s}) = \scale_k(\tmodel(t')\tmodel(t''))$. Because $\tmodel(t')\tmodel(t'') = \tmodel(t't'') = \tmodel(s)$, we have $\lmodel(\varphi_{s}) = \scale_k(\tmodel(s))$. 
        \item If $s = \ReLU(t')$, then $\varphi_s = \left(\varphi_{t'} \ominus \overline{\half}\right) \oplus \overline{\half}$. By the semantics of $\oplus$, $\ominus$ and $\overline{\half}$, we have
        \[
            \begin{aligned}
                \textstyle{\lmodel(\varphi_s)\,} &\textstyle{= \lmodel\left(\left(\varphi_{t'} \ominus \overline{\half}\right) \oplus \overline{\half}\right)} \\
                &\textstyle{= \min\left\{1, \lmodel\left(\varphi_{t'} \ominus \overline{\half}\right) + \half\right\}} \\
                &\textstyle{= \min\left\{1, \max\left\{0, \lmodel\left(\varphi_{t'}\right) - \half\right\} + \half\right\}.}
            \end{aligned}
        \]
        Because $\max\{a, b\} + c = \max\{a+c, b+c\}$ for all $a,b,c \in \R$, we get 
        \[
            \textstyle{\lmodel(\varphi_s) = \min\left\{1, \max\left\{\half, \lmodel\left(\varphi_{t'}\right)\right\}\right\}.}
        \]
        By the induction hypothesis, we have $\lmodel\left(\varphi_{t'}\right) = \scale_k(\tmodel(t'))$ and thus
        \[
            \textstyle{\lmodel(\varphi_s)} \textstyle{= \min\left\{1, \max\left\{\half, \scale_k(\tmodel(t'))\right\}\right\}}.
        \]
        There are two cases to explore: $\tmodel(t') < 0$ and $\tmodel(t') \geq 0$. 
        \begin{enumerate}
            \item If $\tmodel(t') < 0$, then $\scale_k(\tmodel(t')) < \half$ and therefore,
            $\max\left\{\half, \scale_k(\tmodel(t'))\right\} = \half$ and $\lmodel(\varphi_s) = \min\left\{1, \half\right\} = \half$.
            Because $\half = \scale_k(0)$, we have $\lmodel(\varphi_s) = \scale_k(0)$. Because $\tmodel(t') < 0$, we have $\tmodel(s) = \tmodel(\ReLU(t')) = 0$ and $\lmodel(\varphi_s) = \scale_k(\tmodel(s))$.
            \item If $\tmodel(t') \geq 0$, then $\scale_k(\tmodel(t')) \geq \half$ and $\max\left\{\half, \scale_k(\tmodel(t'))\right\} = \scale_k(\tmodel(t'))$, and thus $\lmodel(\varphi_s) = \min\{1,\scale_k(\tmodel(t'))\}$.
            Because $\tmodel(t') \in \left[-k, k\right]$, we know that $\scale_k(\tmodel(t')) \in [0,1]$, which means that $\min\left\{1, \scale_k(\tmodel(t'))\right\} = \scale_k(\tmodel(t'))$ and thus $\lmodel(\varphi_s) = \scale_k(\tmodel(t'))$.
            Finally, because $\tmodel(t') \geq 0$, we must have $\tmodel(s) = \tmodel(\ReLU(t')) = \tmodel(t')$ and therefore $\lmodel(\varphi_s) = \scale_k(\tmodel(s))$.
        \end{enumerate}
    \end{itemize}

    We have now shown the correctness of our construction.
    The claim ``For all $d \in \N$, if $\deg(t) \leq d$, then $\varphi_t$ is a formula of $\cL_{\leq d}$'' now follows from the definition of degree, our construction and Lemmas~\ref{lemma:degrees of formulae} and \ref{lemma: degrees of special formulae} by a simple induction over the structure of $t$.
\end{proof}

In the next lemma, we give an opposite translation from formulae of $\RPLM$ and $\LPHF$ to terms of \coolname.

\begin{lemma}\label{lemma: logic to terms}
    Let $\cL$ be $\RPLM$ or $\LPHF$.
    For each formula $\varphi$ of $\cL$, we can construct an equivalent term $t_{\varphi}$ of \coolname. For all $d \in \N$,
    if $\varphi$ is a formula of $\cL_{\leq d}$,
    then $\deg(t_{\varphi}) \leq d$.
\end{lemma}
\begin{proof}
    The strategy is to recursively build an equivalent term for each formula $\varphi$ of $\cL$.
    We cover the case where $\cL$ is $\RPLM$. The case where $\cL$ is $\LPHF$ follows by Propositions~\ref{proposition: RPLM_d = LPHF_d} and \ref{proposition: RPLM = LPHF}.

    Let $\varphi$ be a formula of $\RPLM$. We will construct a term $t_{\varphi}$ equivalent to $\varphi$ by induction over the structure of $\varphi$ as follows. First, we cover the base cases.
    \begin{itemize}
        \item If $\varphi = \overline{r}$ for some $r \in \Q \cap [0,1]$, then $t_{\varphi} = \overline{r}$.
        \item If $\varphi = p$ for some $p \in \PROP$, then $t_{\varphi} = x_p$.
    \end{itemize}
    Next, assume $\psi$ and $\theta$ are formulae of $\RPLM$, and let $t_{\psi}$ and $t_{\theta}$ be the corresponding terms of \coolname.
    \begin{itemize}
        \item If $\varphi = \psi \impl \theta$, then $t_{\varphi} = \overline{1} - \ReLU(t_{\psi} - t_{\theta})$.
        \item If $\varphi = \psi \odot \theta$, then $t_{\varphi} = t_{\psi} t_{\theta}$.
    \end{itemize}

    Next, we prove the correctness of the given construction. Let $\lmodel$ be a valuation and let $\tmodel$ be an evaluation map such that $\tmodel(x) = \lmodel(p_x)$ for all $x \in \VAR$. We have to show that $\tmodel(t_{\varphi}) = \lmodel(\varphi)$; we show this by induction over the structure of $\varphi$. First, we cover the base cases.
    \begin{itemize}
        \item If $\varphi = \overline{r}$ for some $r \in \Q \cap [0,1]$, then $t_{\varphi} = \overline{r}$ and $\tmodel(t_{\varphi}) = \tmodel(\overline{r}) = r = \lmodel(\overline{r}) = \lmodel(\varphi)$. 
        \item If $\varphi = p$ for some $p \in \PROP$, then $t_{\varphi} = x_p$ and $\tmodel(t_{\varphi}) = \tmodel(x_p) = \lmodel(p) = \lmodel(\varphi)$.
    \end{itemize}
    Next, assume the claim holds for formulae $\psi$ and $\theta$ of $\RPLM$, i.e., $\tmodel(t_{\psi}) = \lmodel(\psi)$ and $\tmodel(t_{\theta}) = \lmodel(\theta)$.
    \begin{itemize}
        \item If $\varphi = \psi \impl \theta$, then $t_{\varphi} = \overline{1} - \ReLU(t_{\psi} - t_{\theta})$. By the semantics of $\ReLU$, we thus have $\tmodel(t_{\varphi}) = 1 - \max\{0, \tmodel(t_{\psi}) - \tmodel(t_{\theta})\}$. By the induction hypothesis, we know it is the case that $\tmodel(t_{\varphi}) = 1 - \max\{0, \lmodel(\psi) - \lmodel(\theta)\}$. Because $-\max\{a,b\} = \min\{-a, -b\}$ for all $a,b \in \R$, we have $\tmodel(t_{\varphi}) = 1 + \min\{0, -\lmodel(\psi) + \lmodel(\theta)\}$. Furhtermore, because $a + \min\{b,c\} = \min\{a+b,a+c\}$ for all $a,b,c \in \R$, $\tmodel(t_{\varphi}) = \min\{1, 1-\lmodel(\psi) + \lmodel(\theta)\}$. Thus by the semantics of $\impl$, we have $\tmodel(t_{\varphi}) = \lmodel(\varphi)$.
        \item If $\varphi = \psi \odot \theta$, then $t_{\varphi} = t_{\psi}t_{\theta}$. Thus $\tmodel(t_{\varphi}) = \tmodel(t_{\psi}t_{\theta}) = \tmodel(t_{\psi})\tmodel(t_{\theta})$. By the induction hypothesis, we have
        $\tmodel(t_{\varphi}) = \lmodel(\psi)\lmodel(\theta)$, and by the semantics of $\odot$, we get $\tmodel(t_{\varphi}) = \lmodel(\varphi)$.
    \end{itemize}

    We have now shown the correctness of our construction.
    The claim ``For all $d \in \N$, if $\varphi$ is a formula of $\cL_{\leq d}$, then $\deg(t_{\varphi}) \leq d$'' follows from Lemma~\ref{lemma:degrees of formulae}, our construction and the definition of degree by a straightforward induction over the structure of $\varphi$.
\end{proof} 

Now, we are ready to prove Theorem~\ref{theorem: R-algebra = RPLM = LPHF}. 

\firsttheorem*
\begin{proof}
    Let $\cL$ be $\RPLM$ or $\LPHF$.
    First, let $(t_1, \dots, t_n)$ be a sequence of terms of \coolname. For each $i \in \R_+$, we need to find a sequence $(\varphi_1, \dots, \varphi_n)$ of formulae of~$\cL$ and some $k \in \R_+$ such that $\varphi_j$ is $(i,k)$-equivalent to $t_j$ for each $j \in \{1, \dots, n\}$. By Lemma~\ref{lemma: terms to logic}, for each $j \in \{1, \dots, n\}$ and for each $i \in \R_+$ there exists some $K_{i,j} \in \Z_+$ such that for each integer $k \geq K_{i,j}$ we can construct a formula that is $(i,k)$-equivalent to $t_j$. Let $K_i = \max\{K_{i,1}, \dots, K_{i,n}\}$. Now for each integer $k \geq K_i$ there exists for each $j \in \{1, \dots, n\}$ a formula $\varphi_j$ that is $(i,k)$-equivalent to $t_j$, and $(\varphi_1, \dots, \varphi_n)$ is the tuple we seek.

    Next, let $(\varphi_1, \dots, \varphi_n)$ be a sequence of formulae of $\cL$. We have to find a sequence $(t_1, \dots, t_n)$ of terms such that for each $j \in \{1, \dots, n\}$, the term $t_j$ and the formula $\varphi_j$ are equivalent. To obtain this sequence, we simply apply Lemma~\ref{lemma: logic to terms} 
    to each formula $\varphi_j$ to obtain the equivalent term $t_j$.
\end{proof}

\section{Proof of Theorem~\ref{theorem: neural networks = RPLM1 = LPHF1}}\label{appendix: Characterising neural networks}

In this section, we give the proof of Theorem~\ref{theorem: neural networks = RPLM1 = LPHF1}. We start by recalling the theorem.

\secondtheorem*

We first establish some auxiliary definitions.
We say that two terms $t$ and $t'$ of \coolname are \defstyle{$[0,1]$-equivalent} if $\tmodel(t) = \tmodel(t')$ for each evaluation map $\tmodel$ such that $\tmodel(x) \in [0,1]$ for each $x \in \VAR$.
For two tuples of terms $T = (t_1, \dots, t_k)$ and $T' = (t_1', \dots, t_k')$ of equal length, we say that $T$ and $T'$ are \defstyle{$[0,1]$-equivalent} if $t_i$ and~$t_i'$ are $[0,1]$-equivalent for each $i \in \{1, \dots, k\}$.

We then establish two auxiliary lemmas concerning neural networks that will be needed in the translations.
The first lemma states that we can arbitrarily increase the depth of a neural network. 

\begin{lemma}\label{lemma: depth increase}
    For each $d \in \N$, each neural network $N$ of depth $d$ and each $d' > d$, there exists a $[0,1]$-equivalent neural network $N'$ of depth $d'$. If $d \neq 0$, then $N$ and $N'$ are equivalent.
\end{lemma}
\begin{proof}
    Let $N = (n_1, \dots, n_k)$ be a neural network of depth $d$. We define $N' := (n_1', \dots, n_{k}')$ where $n_j' := \ReLU(\sum_{i = 1}^{k} \overline{w_{i,j}} n_i + \overline{0})$ where $w_{j,j} = 1$ and $w_{i,j} = 0$ for each $i \neq j$. Now $n_j'$ is equivalent to $\ReLU(n_j)$ and thus $[0,1]$-equivalent to $n_j$ for each $j \in \{1, \dots, k\}$. This is because each neuron $n_j$ is either a variable (if $d = 0$) or of the form $\ReLU(t)$ for some proto-neuron $t$ (if $d \neq 0$). In either case, for each evaluation map $\tmodel$ such that $\tmodel(x) \in [0,1]$ for each $x \in \VAR$, we have $\tmodel(n_j) \geq 0$ and thus $\tmodel(n_j') = \tmodel(\ReLU(n_j)) = \tmodel(n_j)$. Thus, $N'$ is a $[0,1]$-equivalent neural network of depth $d+1$ and we may iterate this process until we reach depth $d'$. If $d \neq 0$, then $n_j$ is of the form $\ReLU(t)$, and thus for each evaluation map~$\tmodel$ we have $\tmodel(n_j') = \tmodel(\ReLU(n_j)) = \tmodel(n_j)$. Thus, $n_j$ and $n_j'$ are equivalent, which means that $N$ and $N'$ are equivalent.
\end{proof}

From the definition of neural networks, we see that a neural network of depth $d$ is built on top of a neural network of depth $d-1$ by adding additional neurons; thus, each neural network of depth $d$ has embedded within it a neural network of each depth $0, \dots, d$. We say that two neural networks of depth at least $d$ are \defstyle{identical up to depth $d$} if their embedded neural networks of depth $d$ are identical.

We next prove another auxiliary lemma, which intuitively shows that any number of neural networks can be modified such that the neural networks are pairwise identical apart from the output neurons. 

\begin{lemma}\label{lemma: combining networks}
    Given $k \in \Z_+$ and neural networks $N_1, \dots, N_k$, there exist neural networks $N_1', \dots, N_k'$ such that $N_i$ and $N_i'$ are $[0,1]$-equivalent for each $i \in \{1, \dots, k\}$, and $N_1', \dots, N_k'$ all have the same depth $d$ and are identical up to depth $d-1$. For each $i \in \{1, \dots, k\}$, if $N_i$ has depth at least $1$, then $N_i$ and $N_i'$ are equivalent.
\end{lemma}
\begin{proof}
    Intuitively, given $t \in \{1, \dots, k\}$, we add every neuron from $N_1, \dots, N_k$ to $N_t'$, except for output neurons which are only added from $N_t$. We copy the weights and biases from $N_1, \dots, N_k$, and we also add the weight $\overline{0}$ between neurons that are not from the same neural network, since our neural networks need to be fully connected. See Figure~\ref{figure: neural network padding} for an illustrated example.

    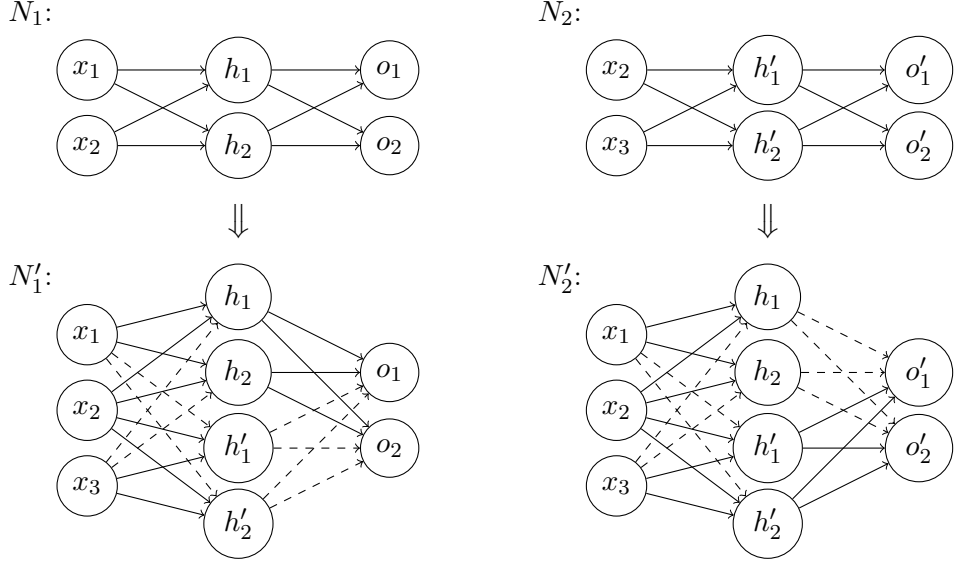
\begin{figure}[ht]
    \centering
    \begin{tikzpicture}
        \node at (-0.75,0.75) () {$N_1$:};
    
        \node[draw, circle] at (0,0) (i1) {$x_1$};
        \node[draw, circle] at (0,-1) (i2) {$x_2$};

        \node[draw, circle] at (2,0) (h1) {$h_1$};
        \node[draw, circle] at (2,-1) (h2) {$h_2$};

        \node[draw, circle] at (4,0) (o1) {$o_1$};
        \node[draw, circle] at (4,-1) (o2) {$o_2$};

        \draw[->] (i1) to (h1);
        \draw[->] (i1) to (h2);
        \draw[->] (i2) to (h1);
        \draw[->] (i2) to (h2);
        \draw[->] (h1) to (o1);
        \draw[->] (h1) to (o2);
        \draw[->] (h2) to (o1);
        \draw[->] (h2) to (o2);

        \node at (6.25,0.75) () {$N_2$:};

        \node[draw, circle] at (7,0) (i1') {$x_2$};
        \node[draw, circle] at (7,-1) (i2') {$x_3$};

        \node[draw, circle] at (9,0) (h1') {$h_1'$};
        \node[draw, circle] at (9,-1) (h2') {$h_2'$};

        \node[draw, circle] at (11,0) (o1') {$o_1'$};
        \node[draw, circle] at (11,-1) (o2') {$o_2'$};

        \draw[->] (i1') to (h1');
        \draw[->] (i1') to (h2');
        \draw[->] (i2') to (h1');
        \draw[->] (i2') to (h2');
        \draw[->] (h1') to (o1');
        \draw[->] (h1') to (o2');
        \draw[->] (h2') to (o1');
        \draw[->] (h2') to (o2');

        \node at (2,-2) () {$\big{\Downarrow}$};

        \node at (-0.75,-2.75) () {$N_1'$:};

        \node[draw, circle] at (0,-3.5) (a1) {$x_1$};
        \node[draw, circle] at (0,-4.5) (a2) {$x_2$};
        \node[draw, circle] at (0,-5.5) (a3) {$x_3$};

        \node[draw, circle] at (2,-3) (b1) {$h_1$};
        \node[draw, circle] at (2,-4) (b2) {$h_2$};
        \node[draw, circle] at (2,-5) (b4) {$h_1'$};
        \node[draw, circle] at (2,-6) (b5) {$h_2'$};

        \node[draw, circle] at (4,-4) (c1) {$o_1$};
        \node[draw, circle] at (4,-5) (c2) {$o_2$};

        \draw[->] (a1) to (b1);
        \draw[->] (a1) to (b2);
        \draw[dashed, ->] (a1) to (b4);
        \draw[dashed, ->] (a1) to (b5);
        \draw[->] (a2) to (b1);
        \draw[->] (a2) to (b2);
        \draw[->] (a2) to (b4);
        \draw[->] (a2) to (b5);
        \draw[dashed, ->] (a3) to (b1);
        \draw[dashed, ->] (a3) to (b2);
        \draw[->] (a3) to (b4);
        \draw[->] (a3) to (b5);
        \draw[->] (b1) to (c1);
        \draw[->] (b1) to (c2);
        \draw[->] (b2) to (c1);
        \draw[->] (b2) to (c2);
        \draw[dashed, ->] (b4) to (c1);
        \draw[dashed, ->] (b4) to (c2);
        \draw[dashed, ->] (b5) to (c1);
        \draw[dashed, ->] (b5) to (c2);

        \node at (9,-2) () {$\big{\Downarrow}$};

        \node at (6.25,-2.75) () {$N_2'$:};

        \node[draw, circle] at (7,-3.5) (d1) {$x_1$};
        \node[draw, circle] at (7,-4.5) (d2) {$x_2$};
        \node[draw, circle] at (7,-5.5) (d3) {$x_3$};

        \node[draw, circle] at (9,-3) (e1) {$h_1$};
        \node[draw, circle] at (9,-4) (e2) {$h_2$};
        \node[draw, circle] at (9,-5) (e4) {$h_1'$};
        \node[draw, circle] at (9,-6) (e5) {$h_2'$};

        \node[draw, circle] at (11,-4) (f1) {$o_1'$};
        \node[draw, circle] at (11,-5) (f2) {$o_2'$};

        \draw[->] (d1) to (e1);
        \draw[->] (d1) to (e2);
        \draw[dashed, ->] (d1) to (e4);
        \draw[dashed, ->] (d1) to (e5);
        \draw[->] (d2) to (e1);
        \draw[->] (d2) to (e2);
        \draw[->] (d2) to (e4);
        \draw[->] (d2) to (e5);
        \draw[dashed, ->] (d3) to (e1);
        \draw[dashed, ->] (d3) to (e2);
        \draw[->] (d3) to (e4);
        \draw[->] (d3) to (e5);
        \draw[dashed, ->] (e1) to (f1);
        \draw[dashed, ->] (e1) to (f2);
        \draw[dashed, ->] (e2) to (f1);
        \draw[dashed, ->] (e2) to (f2);
        \draw[->] (e4) to (f1);
        \draw[->] (e4) to (f2);
        \draw[->] (e5) to (f1);
        \draw[->] (e5) to (f2);
    \end{tikzpicture}
    \caption{Top: Two neural networks $N_1$ and $N_2$. Bottom: Equivalent versions $N_1'$ and $N_2'$ where only the output neurons are different. Dashed lines indicate the weight $0$ between two neurons.}
    \label{figure: neural network padding}
    \end{figure}

    Next, we give the construction formally. 
    First, let $d$ be the maximum depth of $N_1, \dots, N_k$.
    By Lemma~\ref{lemma: depth increase}, 
    we can construct neural networks $M_1, \dots, M_k$ of depth $d$ such that $M_j$ is $[0,1]$-equivalent to $N_j$ for each $j \in \{1, \dots, k\}$.
    Next, we construct the neural network $N_t'$ that is equivalent to $M_t$, starting with the embedded neural network $N_t^0$ of depth $0$ of $N_t'$ and recursively moving to embedded neural networks $N_t^{d'}$ of increasing depths $d'$ until we reach $N_t^d = N_t'$.

    \begin{itemize}
        \item First, we define $N_t^0 := (y_1, \dots, y_{\ell_0})$, where $y_1, \dots, y_{\ell_0}$ are the variables of the input layers of $M_1, \dots, M_k$ with no duplicates (unless $d = 0$, in which case $y_1, \dots, y_{\ell_0}$ are the variables of just $M_t$).
        \item Second, we define $N_t^1$. Let $(m_1, \dots, m_{\ell})$ be the concatenation of the embedded neural networks of depth $1$ of $M_1, \dots, M_k$ (unless $d = 1$ in which case let $(m_1, \dots, m_{\ell})$ be~$M_t$). We define $N_t^{1} = (n_1, \dots, n_{\ell})$, where $n_j = \ReLU\left(\sum_{i = 1}^{\ell_0}\overline{w_{i,j}}y_i + \overline{b_j}\right)$ where $\overline{b_j}$ is the bias of $m_j$ and $\overline{w_{i,j}}$ is the weight from $y_i$ to $m_j$ if they are from the same neural network and otherwise $\overline{w_{i,j}} = \overline{0}$.
        \item Next, we cover the case where $d' > 1$.
        Let $(m_1, \dots, m_\ell)$ be the concatenation of the embedded neural networks of depth $d'$ of $M_1, \dots, M_k$. 
        Let $N_t^{d'} = (n_1, \dots, n_{\ell})$ such that each $n_j$ corresponds to the neuron $m_j$. 
        Let $(m_1', \dots, m_{\ell'}')$ be the concatenation of the embedded neural networks of depth $d'+1$ of $M_1, \dots, M_k$ (unless, again, $d = d'+1$ in which case let $(m_1', \dots, m_{\ell'}')$ be $M_t$). 
        We define $N_t^{d'+1} := (n_1', \dots, n_{\ell'}')$ where $n_j' = \ReLU\left(\sum_{i = 1}^{\ell}\overline{w_{i,j}}n_i + \overline{b_j}\right)$ where $\overline{b_j}$ is the bias of $m_j'$ and $\overline{w_{i,j}}$ is the weight from $m_i$ to $m_j'$ if they are from the same neural network and otherwise $\overline{w_{i,j}} = \overline{0}$.
    \end{itemize}

    Lastly, we obviously define $N_t' := N_t^d$.
    The equivalence of $M_t$ and $N_t'$ is now clear to see, and because $N_t$ and $M_t$ are $[0,1]$-equivalent, $N_t$ and $N_t'$ are thus also $[0,1]$-equivalent. If $N_t$ has depth at least $1$, then $N_t$ and $M_t$ are equivalent, and thus $N_t$ and $N_t'$ are equivalent. With the exception of the output neurons, the construction was identical for each $t$, which means that $N_1', \dots, N_k'$ are pairwise identical up to depth $d-1$.
\end{proof}

Now, we are ready to show the translation from the logics $\RPLM_{\leq 1}$ and $\LPHF_{\leq 1}$ to neurons.

\begin{lemma}\label{lemma: logic to neurons}
    Let $\cL$ be $\RPLM$ or $\LPHF$.
    For each formula $\varphi$ of $\cL_{\leq 1}$, we can construct an equivalent neuron $n_{\varphi}$.
\end{lemma}
\begin{proof}
    We modify the proof of Lemma~\ref{lemma: logic to terms} such that the constructed terms are, in fact, neurons. We cover the case where $\cL$ is $\RPLM$. The case where $\cL$ is $\LPHF$ follows by Propositions~\ref{proposition: RPLM_d = LPHF_d} and \ref{proposition: RPLM = LPHF}.

    Let $\varphi$ be a formula of $\RPLM_{\leq 1}$. We will construct a neuron $n_{\varphi}$ equivalent to $\varphi$ by recursion over the structure of $\varphi$ as follows. First, we cover the base cases.
    \begin{itemize}
        \item If $\varphi = p$ for some $p \in \PROP$, then $n_{\varphi} := x_p$.
        \item If $\varphi$ is a formula of $\RPLM_{\leq 0}$, then there exists some $r \in \Q \cap [0,1]$ such that $\varphi$ is equivalent to~$\overline{r}$. We define $n_{\varphi} := \ReLU\left(\overline{0} \cdot x + \overline{r}\right)$.
    \end{itemize}
    Next, assume that $\psi$ is a formula of $\RPLM_{\leq i}$ and $\theta$ a formula of $\RPLM_{\leq j}$ for some $i,j \in \{0,1\}$, and let $n_{\psi}$ and $n_{\theta}$ be the equivalent neurons. Both $n_{\psi}$ and $n_{\theta}$ can be seen as neural networks with a single output neuron, so by Lemma~\ref{lemma: combining networks} we can construct neurons $n_{\psi}'$ and $n_{\theta}'$ that are $[0,1]$-equivalent to $n_{\psi}$ and $n_{\theta}$ (and thus equivalent to $\psi$ and $\theta$) such that $n_{\psi}'$ and $n_{\theta}'$ have the same depth $d$ and are identical up to depth $d-1$.
    \begin{itemize}
        \item If $i+j \leq 1$ and $\varphi = \psi \odot \theta$, then without loss of generality we may assume that $i = 0$. Because $\psi$ is now a formula of $\RPLM_{\leq 0}$, there exists some $r \in \Q \cap [0,1]$ such that $\psi$ is equivalent to $\overline{r}$. We define $n_{\varphi} := \ReLU(\overline{r} \cdot n_{\theta} + \overline{0})$.
        \item If $\varphi = \psi \impl \theta$, then
        we first define $m_{\varphi} := \ReLU((\overline{1} \cdot n_{\psi}' + \overline{-1} \cdot n_{\theta}') + \overline{0})$, after which we define
        $n_{\varphi} := \ReLU(\overline{- 1} \cdot m_{\varphi} + \overline{1})$.
    \end{itemize}

    Next, we prove the correctness of our construction. Let $\lmodel$ be a valuation, and let $\tmodel$ be the evaluation map such that $\tmodel(x) = \lmodel(p_x)$ for each $x \in \VAR$. We have to show that $\tmodel(n_{\varphi}) = \lmodel(\varphi)$.
    \begin{itemize}
        \item If $\varphi = p$ for some $p \in \PROP$, then $n_{\varphi} = x_p$ and the claim is trivially true.
        \item If $\varphi$ is a formula of $\RPLM_{\leq 0}$, then $n_{\varphi} = \ReLU\left(\overline{0} \cdot x + \overline{r}\right)$ where $r \in \Q \cap [0,1]$ such that $\varphi$ is equivalent to $\overline{r}$. Now $\tmodel(n_{\varphi}) = \max\{0, 0 \cdot \tmodel(x) + r\} = \max\{0,r\}$. Now $r \in \Q \cap [0,1]$ implies $r \geq 0$ and thus $\tmodel(n_{\varphi}) = r = \lmodel(\overline{r})$. Because $\varphi$ is equivalent to~$\overline{r}$, we have $\lmodel(\overline{r}) = \lmodel(\varphi)$ and thus $\tmodel(n_{\varphi}) = \lmodel(\varphi)$.
    \end{itemize}
    Next, assume that the claim holds for formulae $\psi$ of $\RPLM_{\leq i}$ and $\theta$ of $\RPLM_{\leq j}$ where $i,j \in \{0,1\}$, let $n_{\psi}$ and $n_{\theta}$ be the equivalent neurons, and let $n_{\psi}'$ and $n_{\theta}'$ be as above.
    \begin{itemize}
        \item If $i+j \leq 1$ and $\varphi = \psi \odot \theta$, then without loss of generality we may assume that $i = 0$. Then $n_{\varphi} = \ReLU(\overline{r} \cdot n_{\theta} + \overline{0})$ where $r \in \Q \cap [0,1]$ such that $\psi$ is equivalent to $\overline{r}$. Therefore we have $\tmodel(n_{\varphi}) = \max\{0, r\tmodel(n_{\theta}) + 0\} = \max\{0, r\tmodel(n_{\theta})\}$. Because $\psi$ is equivalent to $\overline{r}$, we have $r = \lmodel(\psi)$, and by the induction hypothesis we have $\tmodel(n_{\theta}) = \lmodel(\theta)$, and thus we get $\tmodel(n_{\varphi}) = \max\{0, \lmodel(\psi)\lmodel(\theta)\} = \max\{0, \lmodel(\varphi)\}$. Lastly, because $\lmodel(\varphi) \in [0,1]$, we have $\tmodel(n_{\varphi}) = \lmodel(\varphi)$.
        \item If $\varphi = \psi \impl \theta$, then $n_{\varphi} = \ReLU(\overline{- 1} \cdot m_{\varphi} + \overline{1})$ where $m_{\varphi} = \ReLU((\overline{1} \cdot n_{\psi}' + \overline{-1} \cdot n_{\theta}') + \overline{0})$. Therefore
        $\tmodel(n_{\varphi}) = \max\{0, -\tmodel(m_{\varphi}) + 1\}$ and $\tmodel(m_{\varphi}) = \max\{0, \tmodel(n_{\psi}') - \tmodel(n_{\theta}')\}$. By the induction hypothesis, $\tmodel(n_{\psi}') = \tmodel(n_{\psi}) = \lmodel(\psi)$ and $\tmodel(n_{\theta}') = \tmodel(n_{\theta}) = \lmodel(\theta)$ and therefore $\tmodel(m_{\varphi}) = \max\{0, \lmodel(\psi) - \lmodel(\theta)\}$. Because of the fact that $\lmodel(\psi), \lmodel(\theta) \in [0,1]$ we must have 
        $\tmodel(m_{\varphi}) \leq \max\{0, 1-0\} = 1$ which implies
        $-\tmodel(m_{\varphi}) + 1 \geq -1 + 1 = 0$ and therefore $\tmodel(n_{\varphi}) = - \tmodel(m_{\varphi}) + 1$. Because $-\max\{a,b\} = \min\{-a,-b\}$ for all $a,b \in \R$, $- \tmodel(m_{\varphi}) = \min\{0, - \lmodel(\psi) + \lmodel(\theta)\}$ and thus $\tmodel(n_{\varphi}) = \min\{0, - \lmodel(\psi) + \lmodel(\theta)\} + 1$. Because $\min\{a,b\} + c = \min\{a+c, b+c\}$ for all $a,b,c \in \R$ and by the semantics of~$\impl$, we finally get $\tmodel(n_{\varphi}) = \min\{1,1 - \lmodel(\psi) + \lmodel(\theta)\} = \lmodel(\psi \impl \theta) = \lmodel(\varphi)$.
    \end{itemize}
    This completes the proof of correctness.
\end{proof}

\begin{example}
    Consider the formula $\varphi = \overline{\frac{1}{3}} \impl \left(p \odot \overline{\half}\right)$ of $\RPLM_{\leq 1}$. By the construction from Lemma~\ref{lemma: logic to neurons}, we obtain the equivalent neuron $n_{\varphi}$:
    \[
        \textstyle{
        \ReLU\left(\overline{- 1} \cdot \ReLU\left(\overline{1} \cdot \ReLU\left(\overline{0} x + \overline{0} x_p + \overline{\frac{1}{3}}\right) + \overline{-1} \cdot \ReLU\left(\overline{0} x + \overline{\half} x_p + \overline{0}\right) + \overline{0}\right) + \overline{1}\right).
        }
    \]
    An illustration of the neural network $N := (n_{\varphi})$ is given in Figure~\ref{figure: formula to neuron}. This neural network could be pruned by omitting the variable $x$ from the input layer of the network; this variable is simply a harmless artifact of the translation.

    \begin{figure}[ht]
    \begin{center}
        \begin{tikzpicture}

            \node[draw, circle] at (0,0) (i1) {$x$};
            \node[draw, circle] at (0,-2) (i2) {$x_p$};

            \node[draw, circle] at (3,0) (h1) {$+ \frac{1}{3}$};
            \node[draw, circle] at (3,-2) (h2) {$+ 0$};

            \node[draw, circle] at (6,-1) (h3) {$+0$};

            \node[draw, circle] at (9,-1) (o1) {$+1$};

            \draw[->] (i1) to (h1);
            \draw[->] (i1) to (h2);
            \draw[->] (i2) to (h1);
            \draw[->] (i2) to (h2);
            \draw[->] (h1) to (h3);
            \draw[->] (h2) to (h3);
            \draw[->] (h3) to (o1);
                
            \filldraw[] (2.125,0) node[anchor=south]{$\cdot \, 0$};
            \filldraw[] (2.5,-0.9) node[anchor=south]{$\cdot \, 0$};
            \filldraw[] (2.5,-1.6) node[anchor=south]{$\cdot \, 0$};
            \filldraw[] (2.125,-2) node[anchor=north]{$\cdot \, \half$};

            \filldraw[] (5.125,-0.7) node[anchor=south]{$\cdot \, 1$};
            \filldraw[] (5.125,-1.3) node[anchor=north]{$\cdot \, (-1)$};
                
            \filldraw[] (8,-1) node[anchor=south]{$\cdot \, (-1)$};
        \end{tikzpicture}
        \caption{An illustration of the neural network $N$. Weights are written alongside edges and biases are written inside neurons.}
        \label{figure: formula to neuron}
    \end{center}
\end{figure}
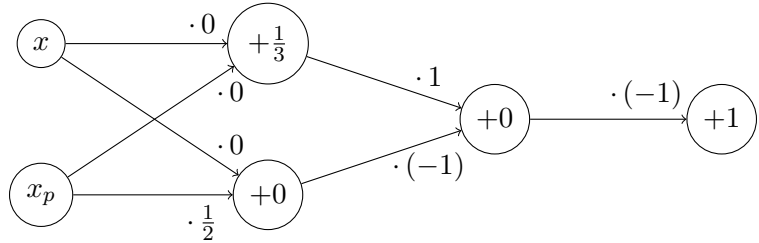
\end{example}

Before delving into the other direction of the characterisation, we first establish a useful lemma.

\begin{lemma}\label{lemma: negation lemma}
    Let $k \in \R_+$, let $\varphi$ be a formula of $\RPLM$ or $\LPH$, let $\lmodel$ be a valuation and let $\ell \in [-k,k]$ such that $\lmodel(\varphi) = \scale_k(\ell)$. Then $\lmodel(\negl \varphi) = \scale_k(-\ell)$.
\end{lemma}
\begin{proof}
    By the semantics of $\negl$, $\lmodel(\negl \varphi) = 1 - \lmodel(\varphi) = 1 - \scale_k(\ell)$. By the definition of $\scale_k$,
    \[
        \textstyle{1 - \scale_k(\ell) = 1 - \frac{k+\ell}{2k} = \frac{2k}{2k} - \frac{k+\ell}{2k} = \frac{k-\ell}{2k} = \scale_k(-\ell),}
    \]
    and thus $\lmodel(\negl \varphi) = \scale_k(-\ell)$.
\end{proof}

Next we will show the other direction of the characterisation, where we translate proto-neurons into the logics $\RPL$ and $\LPHFF$. We split the translation into two parts presented as lemmas. In the first lemma, we replace the rationals of a proto-neuron with integers. In the second lemma, we translate the modified proto-neuron into a formula of $\RPL$ or $\LPHFF$.

\begin{lemma}\label{lemma: neural network inflation}
    For each proto-neuron $t$ there exists some $D \in \Z_+$ such that we can construct a proto-neuron~$s$ that satisfies the following:
    \begin{itemize}
        \item $s$ is equivalent to $\overline{D}t$ and
        \item $s$ contains only integers, i.e., if $s$ contains a subterm $\overline{r}$ for some $r \in \Q$, then $r \in \Z$.
    \end{itemize}
\end{lemma}
\begin{proof}
    To obtain $s$, we intuitively multiply $t$ with an appropriate integer constant $\overline{D}$ (which can be constructively chosen based on the rational constants in $t$) and recursively cancel out the denominators of all the rationals in $t$.
    For example, 
    if $t = \overline{\frac{5}{7}}\left( x + \overline{\frac{3}{7}} \right)$, then we choose $D = 7 \cdot 7 = 49$ and starting from $\overline{49} \, t = \overline{49} \left( \overline{\frac{5}{7}}\left( x + \overline{\frac{3}{7}} \right) \right)$, we obtain the following sequence of proto-neurons that are all equivalent to $\overline{49} \, t$: $s_0 = \overline{5} \left(\overline{7}\left( x + \overline{\frac{3}{7}} \right)\right)$, $s_1 = \overline{5} \left( \overline{7} \, x + \overline{7}\cdot\overline{\frac{3}{7}} \right)$ and $s_2 = \overline{5}\left( \overline{7} \, x + \overline{3} \right)$.

    We now present the general case formally.
    We start with the observation that if $r \in \Q$, then $r = \frac{n_r}{d_r}$ for some numerator $n_r \in \Z$ and denominator $d_r \in \Z_+$ such that $d_r$ is not divisible by $n_r$.

    Next, we show how to construct, for each subterm $s$ of $t$, the appropriate integer $D_s$. First, if $s = \overline{r}$ for some $r \in \Q$, then $D_s := d_r$. If $s = x$ for some $x \in \VAR$, then $D_s := 1$. Next, assume we have defined $D_{s'}$ and $D_{s''}$ for subterms $s'$ and $s''$ of $t$. If $s = s' + s''$, then let~$d$ be the greatest common divisor of $D_{s'}$ and $D_{s''}$, and $D_s := \frac{D_{s'}D_{s''}}{d}$. 
    If $s = s's''$, then $D_s := D_{s'}D_{s''}$. 
    Finally, if $s = \ReLU(s')$, then $D_s := D_{s'}$.

    Let $D$ denote $D_t$.
    We start with the proto-neuron $\overline{D} t$ 
    and modify $\overline{D} t$ into an equivalent proto-neuron $s_0$ as follows.
    \begin{itemize}
        \item If $t = \overline{r}$ for some $r \in \Q$, then $r = \frac{n_r}{d_r}$ and $D = d_r$ and thus we set $s_0 := \overline{n_r}$.
        \item If $t = x$ for some $x \in \VAR$, then $D = 1$ and we set $s_0 := x$.
        \item If $t = t' + t''$ for some proto-neurons $t'$ and $t''$, 
        then $D = \frac{D_{t'}D_{t''}}{d}$ where $d$ is the greatest common divisor of $D_{t'}$ and $D_{t''}$. 
        We set $s_0 := \overline{\frac{D_{t''}}{d}}\overline{D_{t'}}t' + \overline{\frac{D_{t'}}{d}}\overline{D_{t''}}t''$.
        \item If $t = t't''$ for some proto-neurons $t'$ and $t''$, then $D = D_{t'}D_{t''}$ and therefore we define $s_0 := (\overline{D_{t'}}t')(\overline{D_{t''}}t'')$. 
        \item If $t = \ReLU(t')$ for some proto-neuron $t'$, then $D = D_{t'}$ and we set $s_0 := \ReLU(\overline{D_{t'}} t')$.
    \end{itemize}
    In each of the above cases, it is straightforward to show that $s_0$ is indeed equivalent to~$\overline{D}t$. 
    Now in the case of the two first bullets above, we are done. In the case of the last three bullets, we move into the subterm $\overline{D_{t'}}t'$ (and $\overline{D_{t''}}t''$) of $s_0$ and perform the same transformation to obtain a term $s_1$ equivalent to $s_0$.
    Recursively, we thus obtain $s_2$, $s_3$, etc.
    The process only branches in the case of the third and fourth bullets, and in each such case, the further transformations in one branch do not affect the further transformations in the other.
    When this process reaches the atomic level, we will have the proto-neuron $s$ that we seek. A simple induction shows that $s$ is equivalent to $\overline{D}t$ and contains only integers. 
\end{proof}

The next lemma shows that the proto-neuron $s$ constructed in the above lemma can be translated into $\RPL$ and $\LPHFF$.

\begin{lemma}\label{lemma: proto-neuron to RPL and LPHFF}
    Let $\cL$ be $\RPL$ or $\LPHFF$, and let $t$ be a proto-neuron that contains only integers. For each $i \in \R_+$, there exists some $K \in \Z_+$ such that for each rational number
    $k \geq K$ we can construct an $(i,k)$-equivalent formula $\varphi_t$ of $\cL$.
\end{lemma}
\begin{proof}
    Our strategy goes as follows.
    Because $t$ is a proto-neuron, one of the multiplicands in each multiplication does not contain any variables (and thus gets the same value in every evaluation map). Moreover, $t$ contains only integers. Because of these two facts, we can replace $\odot^k$ in the proof of Lemma~\ref{lemma: terms to logic} with the iterated sum operator $\itsum_{n}'$ defined in Appendix~\ref{appendix: auxiliary connectives}.
    We cover the case where $\cL$ is $\RPL$ as the case where $\cL$ is $\LPHFF$ follows by Proposition~\ref{proposition: RPL = LPHFF}.

    Let $i \in \R_+$, define $r_{\max}$, $j$ and $K$ as in the proof of Lemma~\ref{lemma: terms to logic} and let 
    $k \geq K$ be a rational number. We construct a formula $\varphi_{s}$ for each subterm $s$ of $t$ as in the proof of Lemma~\ref{lemma: terms to logic} with the exception of multiplications, where we define instead:
    \begin{itemize}
        \item If $s = t't''$ for some some proto-neurons $t'$ and $t''$, then $\deg(t') = 0$ or $\deg(t'') = 0$.
        \begin{itemize}
            \item If $\deg(t') = \deg(t'') = 0$, then there exists some $z \in \Z$ such that $\tmodel(s) = z$ for each evaluation map $\tmodel$. By the choice of $k$, we have $z \in [-k,k]$ meaning that $\scale_k(z) \in [0,1]$, and because $k,z \in \Q$, we also have $\scale_k(z) = \frac{k+z}{2k} \in \Q$ and thus $\scale_k(z) \in \Q \cap [0,1]$. We define $\varphi_s := \overline{\scale_k(z)}$.
            \item If $\deg(t') \neq 0$ or $\deg(t'') \neq 0$, then by symmetry and without loss of generality, we may assume that $\deg(t'') = 0$. Then there exists some $z \in \Z$ such that $\tmodel(t'') = z$ for each evaluation map $\tmodel$. We define
        \end{itemize}
        \[
            \varphi_s := \begin{cases}
                \itsum'_z \varphi_{t'} &\text{ if } z \in \N \\
                \negl \itsum'_{-z} \varphi_{t'} &\text{ if } z \in \Z \setminus \N.
            \end{cases}
        \] 
    \end{itemize}
    Before proving correctness, we note that in Lemma~\ref{lemma: terms to logic} we had to assume that $k$ is an integer because the connective $\odot^k$ used in the translation is only defined for integers. Here we can relax this assumption because the operator $\itsum'_z$ does not depend on $k$, but we still have to assume that $k$ is rational to ensure that the truth constants used in the translation are defined in $\RPL$.

    We modify the proof of the correctness of the construction from the proof of Lemma~\ref{lemma: terms to logic} as follows. Let $s$ be a subterm of $t$, let $\tmodel$ be an evaluation map such that $\tmodel(x) \in [-i,i]$ for each $x \in \VAR$ and let $\lmodel$ be the unique valuation such that $\lmodel(p) = \scale_k(\tmodel(x_p))$ for each $p \in \PROP$. We have to show that $\lmodel(\varphi_s) = \scale_k(\tmodel(s))$. The proof is again by induction over the structure of $s$ and otherwise identical to the proof of Lemma~\ref{lemma: terms to logic} except for multiplications where we do the following instead. Assume the claim holds for the proto-neurons $t'$ and $t''$, i.e., $\lmodel(\varphi_{t'}) = \scale_k(\tmodel(t'))$ and $\lmodel(\varphi_{t''}) = \scale_k(\tmodel(t''))$.

    Assume $s = t't''$.
    We must show that $\lmodel(\varphi_{s}) = \scale_k(\tmodel(s))$.
    If $\deg(t') = \deg(t'') = 0$, then $\varphi_s = \overline{\scale_k(\tmodel(s))}$ and thus $\lmodel(\varphi_s) = \scale_k(\tmodel(s))$.
    Assume next that $\deg(t') \neq 0$ or $\deg(t'') \neq 0$. By symmetry, we again assume that $\deg(t'') = 0$ and let $z := \tmodel(t'')$ whereby $z \in \Z$.
    We begin by showing that $\tmodel(t') \in \left[-\frac{k}{|z|+1}, \frac{k}{|z|+1}\right]$; by Lemma~\ref{lemma: semantics of iterated plus'}, this will imply that $\lmodel(\itsum'_{|z|} \varphi_{t'}) = \scale_k(|z|\tmodel(t'))$.

    We recall from the proof of Lemma~\ref{lemma: terms to logic} that $|\tmodel(t')| \leq \sqrt{k} = \frac{k}{\sqrt{k}}$. Since $k \geq K \geq j^{2 \length(t)}$, we have $\sqrt{k} \geq \sqrt{j^{2 \length(t)}} = j^{\length(t)}$ and thus $|\tmodel(t')| \leq \frac{k}{j^{\length(t)}}$. Because $j \geq 2$ and $\length(t) > \length(t')$, we have $j^{\length(t)} \geq j^{\length(t')} + 1$ and thus $|\tmodel(t')| \leq \frac{k}{j^{\length(t')} + 1}$. Lastly, we recall from the proof of Lemma~\ref{lemma: terms to logic} that $|z| = |\tmodel(t')| \leq j^{\length(t')}$ and thus $|\tmodel(t')| \leq \frac{k}{|z| + 1}$, i.e., $\tmodel(t') \in \left[-\frac{k}{|z|+1}, \frac{k}{|z|+1}\right]$. 
    Now by Lemma~\ref{lemma: semantics of iterated plus'}, we must have $\lmodel(\itsum'_{|z|} \varphi_{t'}) = \scale_k(|z|\tmodel(t'))$.  We are now ready to cover the two cases.
    \begin{itemize}
        \item If $z \in \N$, then $|z| = z$ and thus by the above observation we get
        \[
            \textstyle{\lmodel(\varphi_s) = \lmodel(\itsum'_z \varphi_{t'}) = \scale_k(z\tmodel(t')) = \scale_k(\tmodel(s)).}
        \]
        \item If $z \in \Z \setminus \N$, then $|z| = -z$. By the above, we get $\lmodel(\itsum'_{-z} \varphi_{t'}) = \scale_k(-z\tmodel(t'))$. Now by Lemma~\ref{lemma: negation lemma}, we get 
        \[
            \textstyle{\lmodel(\varphi_s) = \lmodel(\negl \itsum'_{-z} \varphi_{t'}) = \scale_k(-(-z\tmodel(t'))) = \scale_k(z\tmodel(t')) = \scale_k(\tmodel(s)).}
        \]
    \end{itemize}
    This concludes the proof of the correctness of the construction.
\end{proof}

\begin{example}
    Consider the neural network $N = (o)$ with the single output neuron 
    \[
        o = \ReLU\left(\overline{-8} \cdot \ReLU\left(\overline{4} x + \overline{8} y + \overline{-5}\right) + \overline{2} \cdot \ReLU(\overline{9} x + \overline{7} y + \overline{3}) + \overline{5}\right).
    \]
    An illustration of $N$ is given in Figure~\ref{figure: neuron to formula}.
    Let $i \in \R_+$ and let $k$ be a sufficiently large rational number.
    By the construction from Lemma~\ref{lemma: proto-neuron to RPL and LPHFF}, we obtain the following formula of $\RPL$ that is $(i,k)$-equivalent to the neuron $o$:
    \[
        \begin{aligned}
            \textstyle{\varphi_{o} :=} &\textstyle{\Bigl(\Bigl(\negl\itsum'_{8} \left(\left(\left(\itsum'_4 p_x \oplus' \itsum'_{8} p_y \oplus' \overline{\scale_{k}(-5)}\right) \ominus \overline{\half}\right) \oplus \overline{\half}\right)} \\
            &\textstyle{\oplus' \itsum'_2 \left(\left(\left(\itsum'_{9} p_x \oplus' \itsum'_{7} p_y \oplus' \overline{\scale_k(3)}\right) \ominus \overline{\half}\right) \oplus \overline{\half}\right) \oplus' \overline{\scale_k(5)}\Bigr) \ominus \overline{\half}\Bigr) \oplus \overline{\half}.}
        \end{aligned}
    \]
\end{example}

\begin{figure}[ht]
    \begin{center}
        \begin{tikzpicture}
            \node at (-1,0) () {$N$:};
    
            \node[draw, circle] at (0,0) (i1) {$x$};
            \node[draw, circle] at (0,-2) (i2) {$y$};

            \node[draw, circle] at (3,0) (h1) {$-5$};
            \node[draw, circle] at (3,-2) (h2) {$+ 3$};

            \node[draw, circle] at (6,-1) (o1) {$+5$};

            \draw[->] (i1) to (h1);
            \draw[->] (i1) to (h2);
            \draw[->] (i2) to (h1);
            \draw[->] (i2) to (h2);
            \draw[->] (h1) to (o1);
            \draw[->] (h2) to (o1);

            \filldraw[] (2.125,0) node[anchor=south]{$\cdot \, 4$};
            \filldraw[] (2.5,-0.9) node[anchor=south]{$\cdot \, 8$};
            \filldraw[] (2.5,-1.6) node[anchor=south]{$\cdot \, 9$};
            \filldraw[] (2.125,-2) node[anchor=north]{$\cdot \, 7$};

            \filldraw[] (5.125,-0.7) node[anchor=south]{$\cdot \, (-8)$};
            \filldraw[] (5.125,-1.3) node[anchor=north]{$\cdot \, 2$};
        \end{tikzpicture}
        \caption{An illustration of the neural network $N$. Weights are written along the edges and biases are written inside the neurons.}
        \label{figure: neuron to formula}
    \end{center}
\end{figure}
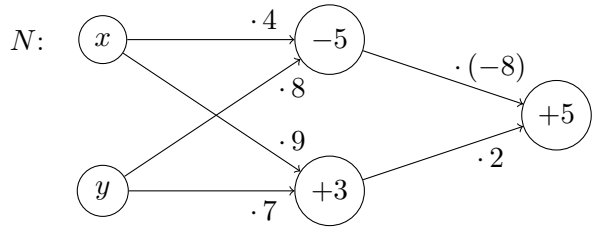

Now we can combine Lemmas~\ref{lemma: neural network inflation} and \ref{lemma: proto-neuron to RPL and LPHFF} to obtain a translation from proto-neurons to formulae of $\RPL$ and $\LPHFF$.

\begin{lemma}\label{lemma: NN-terms to logic}
    Let $\cL$ be $\RPL$ or $\LPHFF$.
    For each proto-neuron and each $i \in \R_+$, there exists some $K \in \Z_+$ such that for each rational number
    $k \geq K$ we can construct an $(i,k)$-equivalent formula of~$\cL$.
\end{lemma}
\begin{proof}
    Let $t$ be a proto-neuron. By Lemma~\ref{lemma: neural network inflation} there exists some $D \in \Z_+$ such that we can construct a 
    proto-neuron $s$ equivalent to $\overline{D}t$ such that $s$ contains only integers. Then by Lemma~\ref{lemma: proto-neuron to RPL and LPHFF}, for each $i \in \R_+$ there exists some $L_i \in \Z_+$ such that for each rational number $\ell \geq L_i$ we can construct an $(i,\ell)$-equivalent formula $\varphi_s$ of $\cL$. 
    
    Now let $i \in \R_+$, let $K := \left\lceil\frac{L_i}{D}\right\rceil$ and let 
    $k \geq K$ be a rational number. Now $kD \geq L_i$ and we can construct a formula $\varphi_s$ of $\cL$ that is $(i, kD)$-equivalent to $s$.
    It remains to show that $\varphi_{s}$ is $(i,k)$-equivalent to $t$. Let $\tmodel$ be an evaluation map such that $\tmodel(x) \in [-i,i]$ for each $x \in \VAR$, and let $\lmodel$ be the unique valuation such that 
    $\lmodel(p) = \scale_{k}(\tmodel(x_p))$ 
    for all $p \in \PROP$. We have to show that 
    $\lmodel(\varphi_{s}) = \scale_{k}(\tmodel(t))$. 
    Because $\varphi_{s}$ is $(i,kD)$-equivalent to $s$, we have $\lmodel(\varphi_{s}) = \scale_{kD}(\tmodel(s))$. Because $s$ is 
    equivalent to $\overline{D}t$
    we have $\tmodel(s) = \tmodel(\overline{D}t) = D\tmodel(t)$ and thus $\lmodel(\varphi_s) = \scale_{kD}(D\tmodel(t))$.
    Now by the definition of $\scale_k$, we get
    \[
        \begin{aligned}
            \textstyle{\lmodel(\varphi_{s})\,} 
            &\textstyle{= \scale_{kD}(D\tmodel(t))} \\
            &\textstyle{= \frac{kD + D\tmodel(t)}{2kD}} \\
            &\textstyle{= \frac{k + \tmodel(t)}{2k}} \\
            &\textstyle{= \scale_k(\tmodel(t)).}
        \end{aligned}
    \]
    Thus, $\varphi_{s}$ is $(i,k)$-equivalent to $t$.

    We note that if we are interested in translating $t$ into a formula of $\RPLM_{\leq 1}$ or $\LPHF_{\leq 1}$ instead, then the application of Lemmas~\ref{lemma: neural network inflation} and \ref{lemma: proto-neuron to RPL and LPHFF} is not necessary (in the case where $k$ is an integer). Instead, by Lemma~\ref{lemma: proto-neuron degrees}, applying Lemma~\ref{lemma: terms to logic} to $t$ would directly result in a formula of $\RPLM_{\leq 1}$ or $\LPHF_{\leq 1}$. 
\end{proof}

Now, we are ready to prove Theorem~\ref{theorem: neural networks = RPLM1 = LPHF1}.

\secondtheorem*
\begin{proof}
    Let $\cL$ be any of the mentioned logics.
    Given a neural network $(n_1, \dots, n_m)$, we can apply the same reasoning as in the proof of Theorem~\ref{theorem: R-algebra = RPLM = LPHF} to find for each $i \in \R_+$ some $K \in \Z_+$ such that for each rational number 
    $k \geq K$ we can construct a tuple $(\varphi_1, \dots, \varphi_m)$ of formulae of $\cL$ such that $\varphi_j$ is $(i,k)$-equivalent to $n_j$ for each $j \in \{1, \dots, m\}$.

    Given a tuple $(\varphi_1, \dots, \varphi_m)$ of formulae of $\cL$, we first use Lemma~\ref{lemma: logic to neurons} to obtain for each $\varphi_i$ an equivalent neuron $n_i$. Now, as each $n_i$ can be interpreted as a neural network with a single output neuron, by Lemma~\ref{lemma: combining networks} we can replace $n_1, \dots, n_m$ with neurons $n_1', \dots, n_m'$ such that each $n_j'$ is $[0,1]$-equivalent to $n_j$ (and thus equivalent to $\varphi_j$), and such that $n_1', \dots, n_m'$ all have the same depth $d$ and are identical up to depth $d-1$. This makes $(n_1', \dots, n_m')$ a neural network.
\end{proof}

\section{Translations between neurons and {\L}ukasiewicz logic}\label{appendix: lukasiewicz logic}

In this section, we will discuss the extent to which {\L}ukasiewicz logic can express neural networks. We will first show some interesting places where our previous constructions fail when restricting to {\L}ukasiewicz logic. Then, we will restrict our attention to a subclass of proto-neurons and the fragment of {\L}ukasiewicz logic obtained by omitting the truth constant $\overline{0}$ from the syntax, and show translations back and forth.
However, the result is not a fuzzy logic characterisation in the sense of Definition~\ref{definition: logical characterisation}, because the translation to logic is more restricted than demanded by the definition.

\subsection{Applying the previous translations to {\L}ukasiewicz logic}

By restricting Lemma~\ref{lemma: logic to neurons} to {\L}ukasiewicz logic,
we obtain the following corollary.

\begin{corollary}\label{corollary: Lukasiewicz logic to neurons}
    For each formula $\varphi$ of {\L}ukasiewicz logic, we can construct an equivalent neuron $n_{\varphi}$. 
\end{corollary}

Next, we will consider the opposite direction and show some interesting places where the recursive construction from Lemmas~\ref{lemma: terms to logic} and \ref{lemma: proto-neuron to RPL and LPHFF} fails. The following proposition shows that the construction fails already in the atomic step for almost all rational constants.

\begin{proposition}\label{proposition: Lukasiewicz logic cannot handle rationals}
    Let $i,k \in \R_+$ such that $i \leq k$ and let $r \in \Q \setminus \{-k,k\}$. There exists no formula of {\L}ukasiewicz logic that is $(i,k)$-equivalent to $\overline{r}$.
\end{proposition}
\begin{proof}
    Theorem 2 in \cite{mcnaughton1951} states that for a function $[0,1]^n \to [0,1]$ to be equivalent to a formula of {\L}ukasiewicz logic, it must be representable by a finite number of polynomials with integer coefficients. In particular, this holds for constant functions. This implies that there is no formula of {\L}ukasiewicz logic that is equivalent to $\overline{r}$ for any $r \in \Q \,\cap\, ]0,1[$, as each polynomial with integer coefficients can only achieve the value $r$ with finitely many inputs. 
    This naturally holds even if we restrict to only those valuations where the truth values of proposition symbols belong to some subinterval $[a,b]$ of $[0,1]$ where $a \neq b$, as is the case with $(i,k)$-equivalence. Thus, for a formula $\varphi$ to be $(i,k)$-equivalent to a term $\overline{r}$ for some $r \in \Q$, we must have $\scale_k(r) \in \{0,1\}$, which implies that $r \in \{-k,k\}$.
\end{proof}

The construction from Lemmas~\ref{lemma: terms to logic} and \ref{lemma: proto-neuron to RPL and LPHFF} also fails in the recursion step. For example, the following proposition shows that even adding variables together is not possible to simulate in {\L}ukasiewicz logic due to the scaling involved.

\begin{proposition}\label{proposition: Lukasiewicz logic cannot handle sum}
    Let $i,k \in \R_+$ such that $i \leq k$ and let $x,y \in \VAR$.
    There exists no formula of {\L}ukasiewicz logic that is $(i,k)$-equivalent to the proto-neuron $x+y$.
\end{proposition}
\begin{proof}
    For the sake of contradiction, assume that $\varphi$ is a formula of {\L}ukasiewicz logic $(i,k)$-equivalent to $x+y$.
    Let $\tmodel$ be an evaluation map such that $\tmodel(z) \in [-i,i]$ for all $z \in \VAR$ and let $\lmodel$ be the valuation such that $\lmodel(p) = \scale_k(\tmodel(x_p))$ for all $p \in \PROP$. Because $\varphi$ is $(i,k)$-equivalent to $x+y$, we must have $\lmodel(\varphi) = \scale_k(\tmodel(x+y)) = \scale_k(\tmodel(x) + \tmodel(y))$.

    Now, let $\psi$ denote the formula obtained from $\varphi$ by replacing each instance of $p_y$ with $\negl p_x$. Now Lemma~\ref{lemma: negation lemma} gives us $\lmodel(\negl p_x) = \scale_k(-\tmodel(x)) = \scale_k(\tmodel(-x))$. This means that $\lmodel(\psi)$ depends on $\tmodel(-x)$ in exactly the same way that $\lmodel(\varphi)$ depends on $\tmodel(y)$, i.e., 
    \[
        \textstyle{\lmodel(\psi) = \scale_k(\tmodel(x) + \tmodel(-x)) = \scale_k(\tmodel(x) - \tmodel(x)) = \scale_k(0) = \scale_k\left(\tmodel\left(\overline{0}\right)\right).}
    \]
    This is a contradiction, because $\psi$ is now a formula of {\L}ukasiewicz logic that is $(i,k)$-equivalent to $\overline{0}$, which is not possible by Proposition~\ref{proposition: Lukasiewicz logic cannot handle rationals}.
\end{proof}

The construction from Lemmas~\ref{lemma: terms to logic} and \ref{lemma: proto-neuron to RPL and LPHFF} also fails when we try to apply the $\ReLU$-function to a simple variable, again due to the scaling involved.

\begin{proposition}\label{proposition: Lukasiewicz logic cannot handle ReLU}
    Let $i,k \in \R_+$ such that $i \leq k$. There exists no formula of {\L}ukasiewicz logic that is $(i,k)$-equivalent to the proto-neuron $\ReLU(x)$.
\end{proposition}
\begin{proof}
    The proof is similar to the proof of Proposition~\ref{proposition: Lukasiewicz logic cannot handle sum}.
    For the sake of contradiction, let $\varphi$ be a formula of {\L}ukasiewicz logic that is $(i,k)$-equivalent to $\ReLU(x)$. 
    Let $\tmodel$ be an evaluation map such that $\tmodel(y) \in [-i,i]$ for all $y \in \VAR$, and let $\lmodel$ be a valuation such that $\lmodel(p) = \scale_k(\tmodel(x_p))$ for all $p \in \PROP$. Because $\varphi$ is $(i,k)$-equivalent to $\ReLU(x)$, we have $\lmodel(\varphi) = \scale_k(\tmodel(\ReLU(x))) = \scale_k(\max\{0,\tmodel(x)\})$.
    
    Now let $\psi$ be the formula obtained from $\varphi$ by replacing each instance of $p_x$ with $\negl\varphi$. By Lemma~\ref{lemma: negation lemma}, $\lmodel(\negl \varphi) = \scale_k(-\tmodel(\ReLU(x))) = \scale_k(\tmodel(-\ReLU(x)))$. Therefore, $\lmodel(\psi)$ depends on $\tmodel(-\ReLU(x))$ in the exact same way as $\lmodel(\varphi)$ depends on $\tmodel(x)$, that is, 
    \[
        \lmodel(\psi) = \scale_k(\max\{0,\tmodel(-\ReLU(x))\}).
    \]
    By the semantics of \coolname and since $-\max\{a,b\} = \min\{-a,-b\}$ for all $a,b \in \R$, we get $\tmodel(-\ReLU(x)) = -\tmodel(\ReLU(x)) = -\max\{0, \tmodel(x)\} = \min\{0, -\tmodel(x)\}$, 
    and therefore we have $\lmodel(\psi) = \scale_k(\max\{0, \min\{0,-\tmodel(x)\}\})$. Now 
    because $\max\{0, \min\{0, a\}\} = 0$ for all $a \in \R$, we have 
    $\lmodel(\psi) = \scale_k(0) = \scale_k\left(\tmodel\left(\overline{0}\right)\right)$.
    This is a contradiction, because now $\psi$ is $(i,k)$-equivalent to $\overline{0}$, which is not possible by Proposition~\ref{proposition: Lukasiewicz logic cannot handle rationals}.
\end{proof}

\subsection{{\L}ukasiewicz logic without any truth constants}

In spite of the fact that the recursive construction from Lemmas~\ref{lemma: terms to logic} and \ref{lemma: proto-neuron to RPL and LPHFF} fails for {\L}ukasiewicz logic, we can nevertheless reverse-engineer which proto-neurons can be expressed in {\L}ukasiewicz logic. %
Let $k \in \Q_+$.\footnote{We assume that $k$ is rational so that the constructed proto-neurons have rational constant symbols as in the most of the paper, but the definitions and results from here on extend in a natural way to the case where $k \in \R_+$.} \defstyle{{\L}ukasiewicz proto-neurons over $k$} are defined as follows.
\begin{itemize}
    \item Each $x \in \VAR$ is a {\L}ukasiewicz proto-neuron over $k$.
    \item If $t$ and $t'$ are {\L}ukasiewicz proto-neurons over $k$, then so is $\overline{k} - \ReLU(t-t')$.
\end{itemize}

\begin{proposition}\label{proposition: what can Lukasiewicz logic do then?}
    Let $k \in \Q_+$. For each {\L}ukasiewicz proto-neuron over $k$ and each $i \in \R_+$ such that $i \leq k$, we can construct an $(i,k)$-equivalent formula of {\L}ukasiewicz logic that does not contain the truth constant $\overline{0}$.
\end{proposition}
\begin{proof}
    We construct the formula recursively over the structure of the {\L}ukasiewicz proto-neuron.

    Let $\tmodel$ be an evaluation map such that $\tmodel(x) \in [-i,i]$ for each $x \in \VAR$ and let $\lmodel$ be the valuation such that $\lmodel(p) = \scale_k(\tmodel(x_p))$ for each $p \in \PROP$. Let $t$ be a {\L}ukasiewicz proto-neuron over $k$. We construct the $(i,k)$-equivalent formula $\varphi_{t}$ by recursion over the structure of $t$.

    \begin{itemize}
        \item If $t = x$ for some $x \in \VAR$, then $\varphi_t := p_x$.
        \item If $t = \overline{k} - \ReLU(t' - t'')$ for some {\L}ukasiewicz proto-neurons $t'$ and $t''$ over $k$, then we define $\varphi_t := \varphi_{t'} \impl \varphi_{t''}$.
    \end{itemize}

    Next, we prove the correctness of the translation. The atomic case where $t = x$ for some $x \in \VAR$ is trivial, so assume that the claim holds for {\L}ukasiewicz proto-neurons $t'$ and $t''$ over $k$ and let $t = \overline{k} - \ReLU(t'-t'')$. 
    It remains to show that $\lmodel(\varphi_t) = \scale_k(\tmodel(t))$.

    By the semantics of $\impl$, we have $\lmodel(\varphi_t) = \min\{1, 1-\lmodel(\varphi_{t'})+\lmodel(\varphi_{t''})\}$. Because $\varphi_{t'}$ and $\varphi_{t''}$ are $(i,k)$-equivalent to $t'$ and $t''$, we have $\lmodel(\varphi_t) = \min\{1, 1-\scale_k(\tmodel(t'))+\scale_k(\tmodel(t''))\}$. By the definition of $\scale_k$, we get 
    \[
        \begin{aligned}
            \textstyle{\lmodel(\varphi_t)\,} &\textstyle{= \min\left\{1, 1 - \frac{k+\tmodel(t')}{2k} + \frac{k+\tmodel(t'')}{2k}\right\}} \\
            &\textstyle{= \min\left\{\frac{2k}{2k}, \frac{2k}{2k} + \frac{-k-\tmodel(t')}{2k} + \frac{k+\tmodel(t'')}{2k}\right\}} \\
            &\textstyle{= \min\left\{\frac{k+k}{2k}, \frac{k+(k-\tmodel(t')+\tmodel(t''))}{2k}\right\}} \\
            &\textstyle{= \min\{\scale_k(k), \scale_k(k-\tmodel(t') + \tmodel(t''))\}.}
        \end{aligned}
    \]
    Because $\scale_k$ is an increasing function, $\min\{\scale_k(a), \scale_k(b)\} = \scale_k(\min\{a,b\})$ for all $a,b \in \R$ and thus $\lmodel(\varphi_t) = \scale_k(\min\{k, k-\tmodel(t') + \tmodel(t'')\})$. Because for all $a,b,c \in \R$ we have $\min\{a+b, a+c\} = a + \min\{b, c\}$, we get $\lmodel(\varphi_t) = \scale_k(k+\min\{0, -\tmodel(t') + \tmodel(t'')\})$. As $\min\{-a, -b\} = -\max\{a, b\}$ for all $a,b\in \R$, $\lmodel(\varphi_t) = \scale_k(k-\max\{0, \tmodel(t') - \tmodel(t'')\})$. Now because of the semantics of $\ReLU$, we have $\tmodel(t) = k-\max\{0, \tmodel(t') - \tmodel(t'')\}$ and thus $\lmodel(\varphi_t) = \scale_k(\tmodel(t))$. This means that $\varphi_t$ is $(i,k)$-equivalent to $t$.
\end{proof}

For {\L}ukasiewicz proto-neurons over $1$, we also obtain a reverse translation, expressed below.

\begin{proposition}\label{proposition: and back}
    For each formula of {\L}ukasiewicz logic that does not contain the truth constant $\overline{0}$, we can construct an equivalent {\L}ukasiewicz proto-neuron over $1$.
\end{proposition}
\begin{proof}
    This follows directly from the proof of Lemma~\ref{lemma: logic to terms}, when we restrict to formulae of {\L}ukasiewicz logic that do not contain the truth constant $\overline{0}$, which are simply built from proposition symbols and the connective $\impl$. 
\end{proof}

Propositions~\ref{proposition: what can Lukasiewicz logic do then?} and \ref{proposition: and back} together constitute a fuzzy logic characterisation of {\L}ukasiewicz proto-neurons over $1$ via the fragment $\cL$ of {\L}ukasiewicz logic obtained by omitting $\overline{0}$ from the grammar. 
However, it is not a characterisation in the sense of Definition~\ref{definition: logical characterisation}, since the definition requires that we obtain an $(i,k)$-equivalent formula for each $i \in \R_+$ while Proposition~\ref{proposition: what can Lukasiewicz logic do then?} only holds when $i \leq k$ and we have $k = 1$. It also remains an open question for which {\L}ukasiewicz proto-neurons over $1$ we can obtain an equivalent neuron.

Lastly, by Propositions~\ref{proposition: what can Lukasiewicz logic do then?} and \ref{proposition: and back}, we obtain the following perhaps amusing corollary.

\begin{corollary}\label{corollary: logic to itself}
    Each formula of {\L}ukasiewicz logic that does not contain the truth constant~$\overline{0}$ is $1$-equivalent to itself.
\end{corollary}
\begin{proof}
    Let $\varphi$ be a formula of {\L}ukasiewicz logic that does not contain the truth constant $\overline{0}$. By Proposition~\ref{proposition: and back}, we can construct a {\L}ukasiewicz proto-neuron $t$ over $1$ that is equivalent to $\varphi$, and by Proposition~\ref{proposition: what can Lukasiewicz logic do then?}, we can construct a formula $\varphi'$ of {\L}ukasiewicz logic that does not contain the truth constant $\overline{0}$ such that $\varphi'$ is $(1,1)$-equivalent to $t$. 
    By Lemma~\ref{lemma: transitive equivalence}, $\varphi'$ is $1$-equivalent to $\varphi$.
    
    It remains to show that $\varphi' = \varphi$ when we construct $t$ and $\varphi'$ according to the proofs of Propositions~\ref{proposition: what can Lukasiewicz logic do then?} and \ref{proposition: and back}.
    We prove the claim by induction over the structure of $\varphi$. If $\varphi = p$ for some $p \in \PROP$, then by the construction from the proof of Lemma~\ref{lemma: logic to terms}, we have $t = x_p$, and by the construction in the proof of Proposition~\ref{proposition: what can Lukasiewicz logic do then?}, we get $\varphi' = p = \varphi$. Next, assume the claim holds for formulae $\psi$ and $\theta$, let $t_{\psi}$ and~$t_{\theta}$ denote the proto-neurons obtained by the construction in the proof of Lemma~\ref{lemma: logic to terms}, and let $\psi'$ and $\theta'$ denote the formulae obtained by the construction in the proof of Proposition~\ref{proposition: what can Lukasiewicz logic do then?}. Let $\varphi = \psi \impl \theta$. Now, by the construction in  Lemma~\ref{lemma: logic to terms}, we get $t = \overline{1} - \ReLU(t_{\psi} - t_{\theta})$. By the construction in the proof of Proposition~\ref{proposition: what can Lukasiewicz logic do then?}, we have $\varphi' = \psi' \impl \theta'$. By the induction hypothesis, we have $\psi' = \psi$ and $\theta' = \theta$ and thus $\varphi' = \psi \impl \theta = \varphi$.
\end{proof}

Corollary~\ref{corollary: logic to itself} shows that each formula of {\L}ukasiewicz logic that does not contain the truth constant $\overline{0}$ can use the interval $\left[\half,1\right]$ to simulate its own behaviour in the interval $[0,1]$.

\end{document}